\documentclass[12pt,reqno,a4paper]{amsart}
\usepackage[left=3cm, top=3.5cm, bottom=3.5cm, right=3cm]{geometry}
\usepackage{mathrsfs}
\usepackage{amsmath,amsxtra,amsfonts,amssymb, amsthm, amscd, epsfig}
\usepackage[dvipsnames,usenames]{color}
\usepackage{amssymb,amsmath,amsfonts,epsfig}
\usepackage{graphicx}
\usepackage{tikz}
\usepackage{hyperref}
\numberwithin{equation}{section}
\usepackage{cite}

\newtheorem{thm}{Theorem}[section]

\newtheorem{cor}[thm]{Corollary}

\newtheorem{lem}[thm]{Lemma}

\newcommand{\eqa}{\begin{eqnarray}}
\newcommand{\eeqa}{\end{eqnarray}}
\newcommand{\beq}{\begin{equation}}
\newcommand{\eeq}{\end{equation}}
\newcommand{\nn}{\nonumber}
\newcommand{\p}{\partial}

\allowdisplaybreaks[1]

\begin{document}

\title[]
{Solutions to Open WDVV Equations for the Universal Whitham Hierarchy}

\author[]{Shilin Ma}

\address[]{Shilin Ma,  School of Mathematics and Statistics,  Xi’an Jiaotong University, Xi’an 710049,  P. R.
	China,}
\email{mashilin@xjtu.edu.cn}
\date{\today}

\begin{abstract}
In this paper, we construct a pair of solutions to the open WDVV equations associated with the infinite-dimensional Frobenius manifolds that underlie the genus-zero universal Whitham hierarchy, and for the resulting flat F-manifolds, we explicitly construct their principal hierarchies. We further demonstrate that this construction is compatible with finite-dimensional reductions, yielding solutions for Frobenius manifolds associated with general rational superpotentials and those subject to a $\mathbb{Z}_2$-symmetry reduction. In particular, the polynomial solutions derived by Basalaev and Buryak via open Saito theory for A and D type singularities are recovered as special cases.

\vskip 2ex

\end{abstract}

\maketitle 

\tableofcontents
\section{Introduction}
The theory of Frobenius manifolds, first introduced by B. Dubrovin \cite{dub1998}, furnishes a geometric framework for the Witten-Dijkgraaf-Verlinde-Verlinde (WDVV) equations of two-dimensional topological field theory (2D-TFT). Briefly, a Frobenius manifold is characterized by a flat metric and a compatible, smoothly varying Frobenius algebra structure on its tangent spaces, complemented by a specific Euler vector field that governs its quasi-homogeneous properties. This mathematical structure has proven to be a crucial tool in diverse areas of mathematical physics, such as Gromov-Witten theory \cite{manin1999frobenius}, singularity theory \cite{hertling2002frobenius}, and, in particular, integrable systems with one spatial variable \cite{dubrovin2001normal}.

It was shown by Dubrovin and Zhang \cite{dubrovin2001normal} that from a semisimple Frobenius manifold, one can associate a dispersionless, tau-symmetric bi-Hamiltonian integrable system on its loop space, called the principal hierarchy. The tau-function of the topological solution to this hierarchy is the genus-zero partition function for the corresponding TFT. The full-genus partition function can be obtained by a quasi-Miura deformation of this hierarchy, which linearizes the Virasoro symmetries acting on the tau-function. This deformation can be obtained by solving the generalized loop equation for the Frobenius manifold. This celebrated result generalizes the well-known Witten conjecture \cite{witten1990two, kontsevich1992intersection}, which states that the partition function of 2D topological quantum gravity is the tau-function of the KdV hierarchy. For generalizations and applications of the Dubrovin-Zhang theory in recent years, see \cite{liu2015bcfg,dubrovin2016hodge,liu2022variational,liu2023variational,liu2025generalized} and references therein.

The concept of infinite-dimensional Frobenius manifolds emerged from efforts to generalize this geometric framework to integrable hierarchies with two spatial variables. The first example was provided by Carlet, Dubrovin, and Mertens\cite{carlet2011infinite}, defining such a manifold on a space of pairs of meromorphic functions. This powerful approach has since been successfully applied to uncover the geometric structures underlying a variety of important hierarchies, including the two-component Kadomtsev-Petviashvili (KP) hierarchy of Type B\cite{wu2012class}, the Toda lattice hierarchy\cite{wu2014infinite}, and an extended KP hierarchy\cite{ma2021infinite}. It should be noted that alternative construction methods also exist, for instance via Schwartz functions \cite{raimondo2012frobenius} or classical r-matrices \cite{szablikowski2015classical}. Furthermore, the corresponding principal hierarchies for some of these manifolds have also been investigated \cite{raimondo2012frobenius,carlet2015principal}.

A fundamental object in the theory of integrable systems is the universal Whitham hierarchy, introduced by Krichever \cite{krichever1988method,krichever1994tau}, primarily due to its "universality"; it contains many celebrated integrable systems such as the dispersionless KP hierarchy as reductions, and finding profound applications in fields ranging from Seiberg-Witten theory to conformal mapping \cite{gorsky1995integrability,itoyama1996integrability,itoyama1997prepotential,krichever1997integrable,gorsky1998rg,krichever2004laplacian,zabrodin2005whitham,wiegmann2000conformal,martinez2006genus}. In a previous work \cite{ma2024infinite}, the present author, together with Wu and Zuo, investigated the genus-zero case of this hierarchy, constructing (following the approach of \cite{carlet2011infinite}) a class of infinite-dimensional Frobenius manifolds and explicitly deriving their corresponding principal hierarchies.

The concept of open WDVV equations was introduced in \cite{horev2012open}, arising from the axioms of open Gromov-Witten theory with applications to the computation of Welschinger invariants. These equations, originally derived to describe generating functions for certain integrals on the moduli spaces of Riemann surfaces with boundaries \cite{pandharipande2024intersection,buryak2022open,buryak2024open}, can be defined for any Frobenius manifold $M$. A crucial insight, first observed by P. Rossi, is the direct correspondence between a solution to these equations and the structure of a flat F-manifold on the product space $M\times U$, where $U$ is a domain in $\mathbb{C}$. This correspondence has proven to be a powerful tool, allowing concepts from the theory of Frobenius manifolds, such as cohomology field theory, Virasoro constraints, Givental theory and integrable systems, to be generalized to the study of open intersection theory \cite{buryak2015equivalence,Buryak:2018ypm,basalaev2019open,arsie2021flat,arsie2023semisimple}.

The aim of this paper is to construct the first solutions to the open WDVV equations for an infinite-dimensional Frobenius manifolds, namely for those associated with the universal Whitham hierarchy \cite{ma2024infinite}. We explicitly construct a pair of such solutions, which in turn define a pair of flat F-manifold structures, and derive the explicit form of their principal hierarchies.

A key feature of our results is their compatibility with finite-dimensional reductions. As established in \cite{ma2024infinite}, both the Frobenius manifold structure and its associated principal hierarchy for the universal Whitham hierarchy restrict under a known limit to those of a Frobenius manifold with rational superpotentials. Moreover, we show that our construction is also compatible with a different restriction to a submanifold underlying the multicomponent BKP hierarchy, and subsequently, with the finite-dimensional limit corresponding to even rational superpotentials. This general result provides a unified framework that encompasses previously known special cases, such as the solutions for Frobenius manifolds associated with A and D type singularities constructed by Basalaev and Buryak via open Saito theory \cite{basalaev2021open}.

The results of this paper may also be viewed as a first step towards a broader goal: to establish a relationship between infinite-dimensional flat F-manifolds, their associated integrable systems, and enumerative geometry, analogous to the celebrated interplay between these structures in the finite-dimensional case \cite{buryak2022open,buryak2024open}.

\subsection{Main results}
To state our main results, we first recall some essential notations from \cite{ma2024infinite}. The manifold $\mathcal{M}$ consists of pairs of meromorphic functions $(a(z),\hat{a}(z))$ defined on an exterior domain $\mathbf{D}^{\text{ext}}$ and an interior domain $\mathbf{D}^{\text{int}}$, respectively. These domains are formed from a set of pairwise non-intersecting closed disks $D_1, \ldots, D_m$ in the Riemann sphere, where $\mathbf{D}^{\text{int}}=\bigcup_{i=1}^m D_i$ and $\mathbf{D}^{\text{ext}}=\mathbb{P}^1\setminus \bigcup_{i=1}^m D_i$. The functions are required to have specific pole structures, with $a(z)$ having a pole at $z=\infty$ and $\hat{a}(z)$ having poles at points $\varphi_j \in D_j$ for each $j=1, \ldots, m$. Their Laurent expansions near these poles are given by:
\begin{align}\label{forma}
	a(z)&=z^{n_{0}}+a_{n_{0}-2}z^{n_{0}-2}+a_{n_{0}-3}z^{n_{0}-3}+\cdots,&\quad &\text{as } z\to \infty; \nn\\	
	\hat{a}(z)&=\hat{a}_{j,-n_{j}}(z-\varphi_{j})^{-n_{j}}+\hat{a}_{j,-n_{j}+1} (z-\varphi_{j})^{-n_{j}+1}+\cdots,& \quad &\text{as } z\to \varphi_{j},
\end{align}
where  $n_0, \ldots, n_m$ are given positive integers. 

The Frobenius manifold structure $(\mathcal{M},\eta,\circ,e,E)$ is described using a decomposition of functions.  Let $\mathcal{H}$ be the space of functions holomorphic in a neighborhood of the domain boundaries $\gamma=\bigcup_{i=1}^m \gamma_i$ where $\gamma_i=\partial D_i$. Any function $f(z)\in\mathcal{H}$ admits a unique decomposition $$f(z)=f(z)_{+}+f(z)_{-},$$ 
where  $f(z)_+$ is analytic in $\mathbf{D}^{\text{int}}$ and $f(z)_-$ is analytic in $\mathbf{D}^{\text{ext}}$, vanishing at infinity. 
For example, $\mathcal{H}$ contains the characteristic functions \( \mathbf{1}_{\gamma_i}(z) \), defined by the condition 
\begin{equation}\label{1fun}
	\mathbf{1}_{\gamma_i}(z)\big|_{\gamma_j} = \delta_{ij}, \quad  i, j \in\{ 1, \ldots, m\}.
\end{equation}
It is a subtle fact that, although  $f(z)\mathbf{1}_{\gamma_{j}}(z)$  vanishes on contours
$\gamma_{j'}$ for $j'\ne j$, its projections $(f(z)\mathbf{1}_{\gamma_{j}}(z))_{\pm}$ may be non-zero on $\gamma_{j'}$.

Using these projections, we define: 
\begin{equation}
	\label{zetaell} \zeta(z) := a(z) - \hat{a}(z), \quad \ell(z) := a(z)_{+} + \hat{a}(z)_{-}. 
	\end{equation} 
	The inverse transformation is given by: 
	\begin{equation} 
		a(z)=\ell(z)+\zeta(z)_{-}, \quad \hat{a}(z)=\ell(z)-\zeta(z)_{+}. 
		\end{equation}
This decomposition is central to the construction, as the flat coordinates $\mathbf{u}$ with respect to the metric $\eta$ are constructed from $\zeta(z)$ and $\ell(z)$. 
The set of these coordinates is denoted by:
\begin{equation}\label{flatcor}
	\mathbf{u}=\{t_{i,s} \mid 1 \leq i \leq m, s \in \mathbb{Z}\} \cup \{h_{0,j} \mid 1 \leq j \leq n_0 - 1\} \cup \{h_{k,r} \mid 1 \leq k \leq m, 0 \leq r \leq n_k\}.
\end{equation}
For the statement of our main theorems, we will also need the following standard truncations of Laurent series:
	\[
\left(\sum_{p} f_p z^p \right)_{\infty,\,\ge0}=\sum_{p\ge0} f_p z^p, \quad
\left(\sum_{p} f_p (z-\varphi_i)^p \right)_{\varphi_i,\,\le -1}=\sum_{p\le -1} f_p (z-\varphi_i)^p.
\]
Furthermore, the integer $d_j$ denotes the winding number of the image of the contour $\gamma_j$ under the map $\zeta(z)$, and the constants $c_p$ are defined by $c_p = \sum_{k=1}^p \frac{1}{k}$ for $p \ge 1$, with $c_0=0$.

With these notations, we now present our main theorems.

\begin{thm}\label{mainthm1}
There exists a pair of functions, $\Omega(\mathbf{u},s)$ on $\mathcal{M}\times \mathbf{D}^{\mathrm{ext}}$ and $\hat{\Omega}(\mathbf{u},s)$ on $\mathcal{M}\times \mathbf{D}^{\mathrm{int}}$, that provide solutions to the open WDVV equations for $\mathcal{M}$. These functions are defined by the following system of differential equations for any flat coordinates $\alpha,\beta\in \mathbf{u}$:
\begin{align*}
	\partial_{\alpha}\partial_{\beta}\Omega(\mathbf{u},s) &= \biggl(\frac{\partial_\alpha \zeta(s) \partial_\beta \zeta(s)}{\zeta'(s)}\biggr)_{\!-} 
	+ \biggl(\frac{\partial_\alpha \ell(s) \partial_\beta \ell(s)}{\ell'(s)}\biggr)_{\!\infty,\,\ge 0} 
	+ \sum_{j=1}^{m}\biggl(\frac{\partial_\alpha \ell(s) \partial_\beta \ell(s)}{\ell'(s)}\biggr)_{\!\varphi_{j},\,\le -1}, \\[1pt]
	\partial_{s}\Omega(\mathbf{u},s) &= a(s), \quad s\in \mathbf{D}^{\mathrm{ext}}; \\[12pt] 
	\partial_{\alpha}\partial_{\beta}\hat{\Omega}(\mathbf{u},s) &= -\biggl(\frac{\partial_\alpha \zeta(s) \partial_\beta \zeta(s)}{\zeta'(s)}\biggr)_{\!+} 
	+ \biggl(\frac{\partial_\alpha \ell(s) \partial_\beta \ell(s)}{\ell'(s)}\biggr)_{\!\infty,\,\ge 0} 
	+ \sum_{j=1}^{m}\biggl(\frac{\partial_\alpha \ell(s) \partial_\beta \ell(s)}{\ell'(s)}\biggr)_{\!\varphi_{j},\,\le -1}, \\[1pt]
	\partial_{s}\hat{\Omega}(\mathbf{u},s) &= \hat{a}(s), \quad s\in \mathbf{D}^{\mathrm{int}}.
\end{align*}
	Here, the projection operators $(\cdot)_{\pm}, (\cdot)_{\infty,\ge0}$ and $(\cdot)_{\varphi_{j},\le-1}$ are understood to act on the Laurent series with respect to $s$.
\end{thm}

The solutions $\Omega(\mathbf{u},s)$ and $\hat{\Omega}(\mathbf{u},s)$ induce a pair of flat F-manifold structures on $\mathcal{M}\times \mathbf{D}^{\mathrm{ext}}$ and $\mathcal{M}\times \mathbf{D}^{\mathrm{int}}$,  respectively. The explicit form of the principal hierarchies for these two structures, which constitutes our second main result, is presented in the following theorem.

\begin{thm}\label{mainthm2}
The principal hierarchies for the flat F-manifolds associated with the solutions  $\Omega(\mathbf{u},s)$ and $\hat{\Omega}(\mathbf{u},s)$ take the following form:
\begin{enumerate}
	\item[(i)]For the case of $\Omega(\mathbf{u},s)$:
\begin{align*}
	\frac{\partial}{\partial \widetilde{T}^{t_{i,l},p}} &= \frac{\partial}{\partial T^{t_{i,l},p}} - \frac{1}{(p+1)!} \frac{d_{i}}{l+d_{i}} \partial_{x}\biggl( \zeta(s)^{l/d_{i}} \bigl(a(s)^{p+1} -\hat{a}(s)^{p+1}\bigr)\mathbf{1}_{\gamma_{i}} \biggr)_{\!-} \partial_{s}, 
	\quad l\ne -d_{i}, \\[12pt]	
	\frac{\partial}{\partial \widetilde{T}^{h_{0,j},p}} &= \frac{\partial}{\partial T^{h_{0,j},p}} + \frac{\Gamma(1-j/n_{0})}{\Gamma(2+p-j/n_{0})} \partial_{x}\Bigl( a(s)^{1+p-j/n_{0}} \Bigr)_{\infty,\ge 0} \partial_{s}, \\[12pt]	
	\frac{\partial}{\partial \widetilde{T}^{h_{k,r},p}} &= \frac{\partial}{\partial T^{h_{k,r},p}} - \frac{\Gamma(1-r/n_{k})}{\Gamma(2+p-r/n_{k})} \partial_{x}\Bigl( \hat{a}(s)^{1+p-r/n_{k}} \Bigr)_{\varphi_{k},\le -1} \partial_{s}, 
	\quad r\ne n_{k}, \\[12pt]	
	\frac{\partial}{\partial \widetilde{T}^{t_{i,-d_{i}},p}} &= \frac{\partial}{\partial T^{t_{i,-d_{i}},p}} 
	- \partial_{x}\Biggl( d_{i} \frac{a^{p}(s)}{p!} \biggl(\log\frac{a(s)^{1/n_{0}}}{s-\varphi_{i}}-\frac{c_{p}}{n_{0}}\biggr) \Biggr)_{\infty,\ge 0} \partial_{s} \nonumber \\[3pt]
	&\quad + \partial_{x}\biggl( d_{i} \frac{a^{p}(s)}{p!} \log\frac{s-\varphi_{i}}{\zeta(s)^{1/d_{i}}}\mathbf{1}_{\gamma_{i}} \biggr)_{\!-} \partial_{s} \nonumber \\[3pt]
	&\quad - \partial_{x}\biggl( d_{i}\frac{a(s)^{p}}{p!} \log(s-\varphi_{i}) \biggr) \partial_{s} \nonumber \\[3pt]
	&\quad + \partial_{x}\biggl( \sum_{j\ne i} d_{i}\frac{a(s)^{p}}{p!} \log(s-\varphi_{i})\mathbf{1}_{\gamma_{j}} \biggr)_{\!-} \partial_{s}, \\[12pt]	
	\frac{\partial}{\partial \widetilde{T}^{h_{k,n_{k}},p}} &= \frac{\partial}{\partial T^{h_{k,n_{k}},p}} 
	- \partial_{x}\biggl( n_{k} \frac{a^{p}(s)}{p!} \log\frac{s-\varphi_{k}}{\zeta(s)^{1/d_{k}}}\mathbf{1}_{\gamma_{k}} \biggr)_{\!-} \partial_{s} \nonumber \\[3pt]
	&\quad + \partial_{x}\biggl( n_{k} \frac{\hat{a}^{p}(s)}{p!} \log\frac{s-\varphi_{k}}{\zeta(s)^{1/d_{k}}}\mathbf{1}_{\gamma_{k}} \biggr)_{\!-} \partial_{s} \nonumber \\[3pt]
	&\quad + \partial_{x}\Biggl( n_{k} \frac{a^{p}(s)}{p!} \biggl(\log\frac{a(s)^{1/n_{0}}}{s-\varphi_{k}}-\frac{c_{p}}{n_{0}}\biggr) \Biggr)_{\infty,\ge 0} \partial_{s} \nonumber \\[3pt]
	&\quad - \partial_{x}\Biggl( n_{k} \frac{\hat{a}^{p}(s)}{p!} \biggl(\log\bigl(\hat{a}(s)^{1/n_{k}}(s-\varphi_{k})\bigr)-\frac{c_{p}}{n_{k}}\biggr) \Biggr)_{\varphi_{k},\le -1} \partial_{s} \nonumber \\[3pt]
	&\quad + \partial_{x}\biggl( n_{k}\frac{a(s)^{p}}{p!} \log(s-\varphi_{k}) \biggr) \partial_{s} \nonumber \\[3pt]
	&\quad - \partial_{x}\biggl( \sum_{j\ne k} n_{k}\frac{a(s)^{p}}{p!} \log(s-\varphi_{k})\mathbf{1}_{\gamma_{j}} \biggr)_{\!-} \partial_{s}, \\[12pt]	
	\frac{\partial}{\partial \widetilde{T}^{s,p-1}} &= \partial_{x}\Bigl( \frac{a(s)^{p}}{p!} \Bigr) \partial_{s}.
\end{align*}
\item[(ii)] For the case of $\hat{\Omega}(\mathbf{u},s)$: 
\begin{align*}
	\frac{\partial}{\partial \widetilde{T}^{t_{i,l},p}} &= \frac{\partial}{\partial T^{t_{i,l},p}} + \frac{1}{(p+1)!} \frac{d_{i}}{l+d_{i}} \partial_{x}\biggl( \zeta(s)^{l/d_{i}} \bigl(a(s)^{p+1} -\hat{a}(s)^{p+1}\bigr)\mathbf{1}_{\gamma_{i}} \biggr)_{\!+} \partial_{s}, 
	\quad l\ne -d_{i}, \\[12pt] 
	\frac{\partial}{\partial \widetilde{T}^{h_{0,j},p}} &= \frac{\partial}{\partial T^{h_{0,j},p}} + \frac{\Gamma(1-j/n_{0})}{\Gamma(2+p-j/n_{0})} \partial_{x}\Bigl( a(s)^{1+p-j/n_{0}} \Bigr)_{\infty,\ge 0} \partial_{s}, \\[12pt]	
	\frac{\partial}{\partial \widetilde{T}^{h_{k,r},p}} &= \frac{\partial}{\partial T^{h_{k,r},p}} - \frac{\Gamma(1-r/n_{k})}{\Gamma(2+p-r/n_{k})} \partial_{x}\Bigl( \hat{a}(s)^{1+p-r/n_{k}} \Bigr)_{\varphi_{k},\le -1} \partial_{s}, 
	\quad r\ne n_{k}, \\[12pt]	
	\frac{\partial}{\partial \widetilde{T}^{t_{i,-d_{i}},p}} &= \frac{\partial}{\partial T^{t_{i,-d_{i}},p}} 
	- \partial_{x}\Biggl( d_{i} \frac{a^{p}(s)}{p!} \biggl(\log\frac{a(s)^{1/n_{0}}}{s-\varphi_{i}}-\frac{c_{p}}{n_{0}}\biggr) \Biggr)_{\infty,\ge 0} \partial_{s} \nonumber \\[3pt] 
	&\quad - \partial_{x}\biggl( d_{i} \frac{a^{p}(s)}{p!} \log\frac{s-\varphi_{i}}{\zeta(s)^{1/d_{i}}}\mathbf{1}_{\gamma_{i}} \biggr)_{\!+} \partial_{s} \nonumber \\[3pt]
	&\quad - \partial_{x}\biggl( \sum_{j\ne i} d_{i}\frac{a(s)^{p}}{p!} \log(s-\varphi_{i})\mathbf{1}_{\gamma_{j}} \biggr)_{\!+} \partial_{s}, \\[12pt] 
	\frac{\partial}{\partial \widetilde{T}^{h_{k,n_{k}},p}} &= \frac{\partial}{\partial T^{h_{k,n_{k}},p}} 
	+ \partial_{x}\biggl( n_{k} \frac{a^{p}(s)}{p!} \log\frac{s-\varphi_{k}}{\zeta(s)^{1/d_{k}}}\mathbf{1}_{\gamma_{k}} \biggr)_{\!+} \partial_{s} \nonumber \\[3pt] 
	&\quad - \partial_{x}\biggl( n_{k} \frac{\hat{a}^{p}(s)}{p!} \log\frac{s-\varphi_{k}}{\zeta(s)^{1/d_{k}}}\mathbf{1}_{\gamma_{k}} \biggr)_{\!+} \partial_{s} \nonumber \\[3pt]
	&\quad + \partial_{x}\Biggl( n_{k} \frac{a^{p}(s)}{p!} \biggl(\log\frac{a(s)^{1/n_{0}}}{s-\varphi_{k}}-\frac{c_{p}}{n_{0}}\biggr) \Biggr)_{\infty,\ge 0} \partial_{s} \nonumber \\[3pt]
	&\quad - \partial_{x}\Biggl( n_{k} \frac{\hat{a}^{p}(s)}{p!} \biggl(\log\bigl(\hat{a}(s)^{1/n_{k}}(s-\varphi_{k})\bigr)-\frac{c_{p}}{n_{k}}\biggr) \Biggr)_{\!\varphi_{k},\le -1} \partial_{s} \nonumber \\[3pt]
	&\quad + \partial_{x}\biggl( n_{k}\frac{\hat{a}(s)^{p}}{p!} \log(s-\varphi_{k})\mathbf{1}_{\gamma_{k}} \biggr) \partial_{s} \nonumber \\[3pt]
	&\quad + \partial_{x}\biggl( \sum_{j\ne k} n_{k}\frac{a(s)^{p}}{p!} \log(s-\varphi_{k})\mathbf{1}_{\gamma_{j}} \biggr)_{\!+} \partial_{s}, \\[12pt]	
	\frac{\partial}{\partial \widetilde{T}^{s,p-1}} &= \partial_{x}\Bigl( \frac{\hat{a}(s)^{p}}{p!} \Bigr) \partial_{s},
\end{align*}
\end{enumerate}
where $\frac{\partial}{\partial T^{u,p}}$ denotes the flow of the principal hierarchy for $\mathcal{M}$ corresponding to the coordinate $u \in \mathbf{u}$, as constructed in \cite{ma2024infinite}.
\end{thm}

We now consider the finite-dimensional reduction of our construction. In the limit $\zeta(z)\to 0$, the infinite-dimensional Frobenius manifold $\mathcal{M}$ degenerates to a finite-dimensional Frobenius manifold (denoted $M$) defined on the space of rational functions of the form:
\[
\ell(z)= z^{n_0} + a_{n_0-2}z^{n_0-2} + \cdots + a_1 z + a_{0} + \sum_{i=1}^{m}\sum_{j=1}^{n_i}b_{i,j}(z-\varphi_{i})^{-j}.
\]
The flat coordinates for $M$ are denoted by $\mathbf{h}$, which correspond to the subset of $\mathbf{u}$ given by:
\[
\mathbf{h} = \{h_{0,j} \mid 1 \leq j \leq n_0 - 1\} \cup \{h_{k,r} \mid 1 \leq k \leq m, 0 \leq r \leq n_k\}.
\]
The principal hierarchy for this manifold corresponds to a subhierarchy of the principal hierarchy for $\mathcal{M}$, and coincides with the dispersionless limit of the constrained KP hierarchy for $m=1$ \cite{aoyama1996topological,liu2015central}. Our construction is compatible with this reduction, yielding the following corollary.

\begin{cor}\label{maincor4}
	In the limit $\zeta(z)\to 0$, the solutions $\Omega$ and $\hat{\Omega}$ both reduce to a single function $F^{o}(\mathbf{h},s)$, defined by the system
	\begin{align*}
		\partial_{\alpha}\partial_{\beta}F^{o}(\mathbf{h},s)&= \left(\frac{\partial_\alpha \ell(s) \partial_\beta \ell(s)}{\ell^{\prime}(s)}\right)_{\infty,\,\ge 0}+\sum_{j=1}^{m}\left(\frac{\partial_\alpha \ell(s) \partial_\beta \ell(s)}{\ell^{\prime}(s)}\right)_{\varphi_{j},\,\le -1}, \\
		\partial_{s}F^{o}(\mathbf{h},s)&=\ell(s),
	\end{align*}
	which provides a solution to the open WDVV equations for the Frobenius manifold $M$.

Moreover, the principal hierarchy for the original flat F-manifold restricts to the one associated with the F-manifold of $F^{o}(\mathbf{h},s)$, which takes the form:
\begin{align*}
	\frac{\partial}{\partial \widetilde{T}^{h_{0,j},p}} &= \frac{\partial}{\partial T^{h_{0,j},p}} + \frac{\Gamma(1-j/n_{0})}{\Gamma(2+p-j/n_{0})} \partial_{x}\Bigl( \ell(s)^{1+p-j/n_{0}} \Bigr)_{\infty,\ge 0} \partial_{s}, \\[12pt]
	\frac{\partial}{\partial \widetilde{T}^{h_{k,r},p}} &= \frac{\partial}{\partial T^{h_{k,r},p}} - \frac{\Gamma(1-r/n_{k})}{\Gamma(2+p-r/n_{k})} \partial_{x}\Bigl( \ell(s)^{1+p-r/n_{k}} \Bigr)_{\varphi_{k},\le -1} \partial_{s}, 
	\quad r\ne n_{k}, \\[12pt]
	\frac{\partial}{\partial \widetilde{T}^{h_{k,n_{k}},p}} &= \frac{\partial}{\partial T^{h_{k,n_{k}},p}} 
	+ \partial_{x}\Biggl( n_{k} \frac{\ell(s)^{p}}{p!} \biggl(\log\biggl(\frac{\ell(s)^{1/n_{0}}}{s-\varphi_{k}}\biggr)-\frac{c_{p}}{n_{0}}\biggr) \Biggr)_{\infty,\ge 0} \partial_{s} \nonumber \\[3pt]
	&\quad - \partial_{x}\Biggl( n_{k} \frac{\ell(s)^{p}}{p!} \biggl(\log\bigl(\ell(s)^{1/n_{k}}(s-\varphi_{k})\bigr)-\frac{c_{p}}{n_{k}}\biggr) \Biggr)_{\varphi_{k},\le -1} \partial_{s} \nonumber \\[3pt]
	&\quad + \partial_{x}\biggl( n_{k}\frac{\ell(s)^{p}}{p!} \log(s-\varphi_{k}) \biggr) \partial_{s} \nonumber \\[3pt]
	&\quad - \partial_{x}\biggl( \sum_{j\ne k} n_{k}\frac{\ell(s)^{p}}{p!} \log(s-\varphi_{k}) \biggr)_{\varphi_{j},\le -1} \partial_{s}, \\[12pt]
	\frac{\partial}{\partial \widetilde{T}^{s,p-1}} &= \partial_{x}\Bigl( \frac{\ell(s)^{p}}{p!} \Bigr) \partial_{s},
\end{align*}
where $\frac{\partial}{\partial T^{u,p}}$ denotes the flow of the principal hierarchy for $M$ corresponding to $u \in \mathbf{h}$.
\end{cor}

In the special case $m=0$, our solution for $M$ coincides with the polynomial solution for the A-type singularity constructed by Basalaev and Buryak \cite{basalaev2021open}. Their construction, based on open Saito theory, yields a solution whose series expansion in the variable $s$ reveals the transition functions between the flat and unfolding coordinates of the underlying Frobenius manifold. Furthermore, this solution is known to coincide with the generating function for intersection numbers in open r-spin theory \cite{buryak2024open}. In particular, for the case where $m=0$ and $n_{0}=2$, the principal hierarchy for our solution coincides with the dispersionless limit of the Burgers-KdV hierarchy \cite{buryak2015equivalence}, whose tau-tuple was studied in \cite{yang2021extension}.

Another type of reduction involves restricting to the fixed points of $\mathcal{M}$ under the involution $\iota$ on $\mathbb{P}^{1}$ defined by $z \to -z$. We assume the parameters and domains satisfy the symmetry conditions:
\[
m=2m'+1,\quad n_{0}=2n_{0}',\quad n_{1}=2n_{1}',\quad n_{2j-2}= n_{2j-1}=n_{j}',\quad 2\le j\le m'+1,
\]
and
\[
\iota(D_{1})=D_{1}, \quad \iota(D_{2j-2})=D_{2j-1},\quad 2\le j\le m'+1.
\]
With these assumptions, we denote by $\hat{\mathcal{M}}$ the submanifold of $\mathcal{M}$ consisting of the fixed points of $\iota$, that is, points $\vec{a}=(a(z),\hat{a}(z))\in \mathcal{M}$ satisfying:
\[
a(z)=a(-z),\quad\hat{a}(z)=\hat{a}(-z).
\]
The metric $\langle \cdot, \cdot \rangle_{\eta}$ on $\mathcal{M}$ can be restricted to $\hat{\mathcal{M}}$, which has a set of flat coordinates
\begin{align*}
	\hat{\mathbf{u}}=\hat{\mathbf{t}}\cup\hat{\mathbf{h}},
\end{align*}
where
\begin{align*}
	\hat{\mathbf{t}} &= \{t_{1,2l-1} \mid l\in \mathbb{Z}\} \cup \{t_{2i-2,l} \mid 2\le i\le m'+1, l\in\mathbb{Z}\}; \\
	\hat{\mathbf{h}} &= \{h_{0,2j-1} \mid 1\le j\le n_{0}'\} \cup \{h_{1,2j-1} \mid 1\le j\le n_{1}'\} \\
	&\quad \cup \{h_{2k-2,r} \mid 2\le k\le m'+1, 0\le r\le n_{k}'\}.
\end{align*}

\begin{cor}\label{maincor5}
	The Frobenius manifold structure on $\mathcal{M}$, as well as the subhierarchy defined by the flows:
	 $$
	 \{\frac{\p}{\p T^{t_{1,2l-1}}}\}_{s\in\mathbb{Z}}\cup\{\frac{\p}{\p T^{t_{2i-2,l}}}-\frac{\p}{\p T^{t_{2i-1,l}}}\}_{2\le i\le m'+1;l\in\mathbb{Z}}
	 $$
	 and
	 $$
	 \{\frac{\p}{\p T^{h_{0,2j-1}}}\}_{1\le j\le n_{0}'}\cup \{\frac{\p}{\p T^{h_{1,2j-1}}}\}_{1\le j\le n_{1}'}\cup \{\frac{\p}{\p T^{h_{2k-2,r}}}-\frac{\p}{\p T^{h_{2k-1,r}}}\}_{2\le k\le m'+1;0\le r\le n_{k}'}
	 $$
	can be restricted to $\hat{\mathcal{M}}$, constituting its Frobenius manifold structure and principal hierarchy. 
	
	Furthermore, the solutions $\Omega$ and $\hat{\Omega}$ remain solutions to the open WDVV equations upon restriction to $\hat{\mathcal{M}}$. The associated principal hierarchies are then given by the restriction of the corresponding flows (defined on $\mathcal{M}\times \mathbf{D}^{\mathrm{ext}}$ and $\mathcal{M}\times \mathbf{D}^{\mathrm{int}}$)
	 to $\hat{\mathcal{M}}\times \mathbf{D}^{\mathrm{ext}}$ and $\hat{\mathcal{M}}\times \mathbf{D}^{\mathrm{int}}$, respectively.
\end{cor}
The Frobenius manifold structure on $\hat{\mathcal{M}}$, which is a generalization of that underlying the two-component BKP hierarchy \cite{wu2012class} (corresponding to the case $m'=0$), degenerates in the limit $\zeta(z)\to 0$ to the corresponding structure on the space $\hat{M}$ of rational functions of the form:
\[
\ell(z)=z^{2n_{0}'}+\sum_{j=0}^{n_{0}'-1}b_{0,j}z^{2j}+\sum_{j=1}^{n_{1}'}b_{1,j}z^{-2j}+\sum_{i=2}^{m'+1}\sum_{j=1}^{n_{i}'}b_{i,j}(z^2-b_{i,0})^{-j}.
\]
This structure coincides with that on the orbit space of the Coxeter group of D type \cite{zuo2007frobenius}, which corresponds to the case $m'=0$ and $n'_{1}=1$. Furthermore, the principal hierarchy for this case is the dispersionless limit of the Drinfeld-Sokolov hierarchy of D type \cite{drinfel1985lie,liu2011drinfeld}.

\begin{cor}\label{maincor6}
	In the limit $\zeta(z)\to 0$, the solutions $\Omega$ and $\hat{\Omega}$ (which remain solutions for $\hat{\mathcal{M}}$) both reduce to a single function $F^{o}(\hat{\mathbf{h}},s)$, which provides a solution to the open WDVV equations for $\hat{M}$. Furthermore, the principal hierarchy for the associated flat F-manifold is obtained by restricting the corresponding flows
	\[
	\left\{\frac{\partial}{\partial \widetilde{T}^{h_{0,2j-1}}}\right\}_{1\le j\le n_{0}'}\cup \left\{\frac{\partial}{\partial \widetilde{T}^{h_{1,2j-1}}}\right\}_{1\le j\le n_{1}'}\cup \left\{\frac{\partial}{\partial \widetilde{T}^{h_{2k-2,r}}}-\frac{\partial}{\partial \widetilde{T}^{h_{2k-1,r}}}\right\}_{2\le k\le m'+1;\, 0\le r\le n_{k}'}
	\]
	(defined on $\mathcal{M}\times \mathbf{D}^{\mathrm{ext}}$ or $\mathcal{M}\times \mathbf{D}^{\mathrm{int}}$) to $\hat{\mathcal{M}}$ and subsequently taking the limit $\zeta(z) \to 0$.
\end{cor}
In the special case where $m'=0$ and $n'_{1}=1$, our solution for $\hat{M}$ coincides with the polynomial solution for the D-type singularity constructed by Basalaev and Buryak \cite{basalaev2021open}. We hope that our result will inspire the construction of a dispersive hierarchy for this case, analogous to the A-type case, or more generally, an open-type Drinfeld-Sokolov hierarchy.

This paper is organized as follows. In Section 2, we review the essential aspects of the infinite-dimensional Frobenius manifold structure on $\mathcal{M}$ that are necessary for the proofs of our main theorems. Section 3 is devoted to the proofs of Theorem \ref{mainthm1} and Theorem \ref{mainthm2}. In Section 4, we study the reductions of our construction, proving Corollaries \ref{maincor4}--\ref{maincor6}.

\section{Frobenius manifold structure on $\mathcal{M}$}
In this section, we recall the construction of the infinite-dimensional Frobenius manifold structure on $\mathcal{M}$ introduced in \cite{ma2024infinite}, focusing on the explicit formulas that are essential for the proofs of our main theorems.

\subsection{Invariant metric}
Given a point $\vec{a}=(a(z),\hat{a}(z))\in\mathcal{M}$, we identify any vector $X$ in the tangent space $T_{\vec{a}}\mathcal{M}$ with the derivative $(\partial_{X} a(z), \partial_{X}\hat{a}(z))$. The invariant metric on $\mathcal{M}$ is defined by
\begin{equation}\label{whimetric}
	\langle X_1, X_2 \rangle_{\eta} =
	-\frac{1}{2\pi\mathrm{i}} \sum_{j=1}^{m} \oint_{\gamma_{j}} \frac{\partial_{1}\zeta(z) \partial_{2}\zeta(z)}{\zeta'(z)} \, dz
	-\left(\operatorname*{Res}_{z=\infty} + \sum_{j=1}^{m} \operatorname*{Res}_{z=\varphi_{j}}\right) \frac{\partial_{1}\ell(z) \partial_{2}\ell(z)}{\ell'(z)} \, dz,
\end{equation}
where $\partial_{\nu}=\partial_{X_{\nu}}$ for $\nu=1,2$, and the components $\zeta(z), \ell(z)$ are as defined in \eqref{zetaell}.

The dual description of the metric is given by a linear map \( \eta: \mathcal{H} \times \mathcal{H} \to T_{\vec{a}}\mathcal{M} \) as
\begin{align}
	\eta(\omega(z), \hat{\omega}(z))=&\bigl(a'(z)(\omega(z)+\hat{\omega}(z))_{-}-(\omega(z) a'(z)+\hat{\omega}(z)\hat{a}'(z))_{-},\nn\\
	&-\hat{a}'(z)(\omega(z)+\hat{\omega}(z))_{+}+(\omega(z) a'(z)+\hat{\omega}(z)\hat{a}'(z))_{+}\bigr), \label{etaom}
\end{align}
which satisfies
\begin{equation}\label{pairmec}
	\langle\vec{\omega},X\rangle=\langle\eta (\vec{\omega}),X\rangle_{\eta},
\end{equation}
where
\begin{equation}\label{whipair}
	\langle \vec{\omega}, X \rangle := \frac{1}{2\pi \mathrm{i}} \sum_{j=1}^{m} \oint_{\gamma_j} \left( \omega(z) \partial_X a(z) + \hat{\omega}(z) \partial_X \hat{a}(z) \right) \, dz.
\end{equation}
The surjectivity of the map $\eta$, established in \cite{ma2024infinite}, allows for the identification of the cotangent space with the quotient space:
\begin{equation}\label{cotspace}
	T^{\ast}_{\vec{a}}\mathcal{M}=( \mathcal{H} \times \mathcal{H} )/\ker \eta.
\end{equation}

\subsection{Flat coordinates}
The flatness of the metric $\eta$ was established by explicitly constructing the flat coordinates $\mathbf{u}$ (see \eqref{flatcor}), which are defined as the coefficients of the Laurent expansions of the inverse functions of $\zeta(z)$ and $\ell(z)$, as follows:
\begin{align*}
	z &= \sum_{s \in \mathbb{Z}} t_{j,s} \zeta(z)^{s/d_{j}}, \quad z \in \gamma_j; \\
	z &= \ell(z)^{1/n_0} - h_{0,1} \ell(z)^{-1/n_0} - h_{0,2} \ell(z)^{-2/n_0} - \cdots, \quad z \rightarrow \infty; \\
	z &=  h_{j,0} + h_{j,1} \ell(z)^{-1/n_j} + h_{j,2} \ell(z)^{-2/n_j} + \cdots, \quad z \rightarrow \varphi_{j}.
\end{align*}
The tangent vectors corresponding to these coordinates can be computed explicitly:
\begin{align*}
	\frac{\partial \vec{a}}{\partial t_{i,s}} &= \left(-(\zeta(z)^{s/d_{i}}\zeta'(z)\mathbf{1}_{\gamma_{i}})_{-}, (\zeta(z)^{s/d_{i}}\zeta'(z)\mathbf{1}_{\gamma_{i}})_{+}\right),\quad 1\le i\le m, ~ s\in\mathbb{Z}; \\
	\frac{\partial \vec{a}}{\partial h_{0,j}} &= \left((\ell'(z)\ell(z)^{-j/n_0})_{\infty,\,\ge 0}, (\ell'(z)\ell(z)^{-j/n_0})_{\infty,\,\ge 0}\right),\quad 1\le j\le n_{0}-1; \\
	\frac{\partial \vec{a}}{\partial h_{k,r}} &= \left(-(\ell'(z)\ell(z)^{-r/n_k})_{\varphi_{k},\,\le -1}, -(\ell'(z)\ell(z)^{-r/n_k})_{\varphi_{k},\,\le -1}\right),\quad 1\le k\le m, ~ 0\le r\le n_{k}.
\end{align*}
The only non-vanishing components of the metric in this basis are:
\begin{align*}
	\left\langle\frac{\partial}{\partial t_{i,s}}, \frac{\partial}{\partial t_{i,s'}}\right\rangle_\eta &= -d_{i} \delta_{-d_{i}, s+s'}, \\
	\left\langle\frac{\partial}{\partial h_{0,j}}, \frac{\partial}{\partial h_{0,j'}}\right\rangle_\eta &= n_{0} \delta_{n_{0}, j+j'}, \\
	\left\langle\frac{\partial}{\partial h_{k,r}}, \frac{\partial}{\partial h_{k,r'}}\right\rangle_\eta &= n_{k} \delta_{n_{k}, r+r'}.
\end{align*}

\subsection{Levi-Civita connection}
The Levi-Civita connection $\nabla$ associated with the metric $\langle \cdot, \cdot \rangle_{\eta}$ is given by the following expression for the covariant derivative of a vector field \(\partial_2\) along \(\partial_1\):
\begin{align} \label{conn}
	(\nabla_{\partial_1} \partial_2) (\vec{a}) = \Biggl(
	& \partial_1 \partial_2 a(z) - \biggl(\frac{\partial_1 \zeta(z) \partial_2 \zeta(z)}{\zeta'(z)}\biggr)_{\!-}' 
	- \biggl(\frac{\partial_1 \ell(z) \partial_2 \ell(z)}{\ell'(z)}\biggr)_{\!\infty,\,\ge 0}' \nonumber \\
	&\quad - \sum_{j=1}^{m}\biggl(\frac{\partial_1 \ell(z) \partial_2 \ell(z)}{\ell'(z)}\biggr)_{\!\varphi_{j},\,\le -1}', \quad \nonumber \\[8pt]
	& \partial_1 \partial_2 \hat{a}(z) + \biggl(\frac{\partial_1 \zeta(z) \partial_2 \zeta(z)}{\zeta'(z)}\biggr)_{\!+}' 
	- \biggl(\frac{\partial_1 \ell(z) \partial_2 \ell(z)}{\ell'(z)}\biggr)_{\!\infty,\,\ge 0}' \nonumber \\
	&\quad - \sum_{j=1}^{m}\biggl(\frac{\partial_1 \ell(z) \partial_2 \ell(z)}{\ell'(z)}\biggr)_{\!\varphi_{j},\,\le -1}' \Biggr).
\end{align}
\subsection{Multiplication structure}\label{secmul}
The Frobenius manifold structure includes a commutative and associative multiplication $\circ$ on each tangent space $T_{\vec{a}}\mathcal{M}$, which is constructed in two steps. First, a commutative and associative multiplication is defined on the ambient space $\mathcal{H} \times \mathcal{H}$ by the formula:
\begin{align*}
	\vec{\omega}_1 \circ \vec{\omega}_{2} = \Bigl(
	&\omega_{2}(\omega_{1}a')_{+} - \omega_{1}(\omega_{2}a')_{-} - \omega_{2}(\hat{\omega}_{1}\hat{a}')_{-} - \omega_{1}(\hat{\omega}_{2}\hat{a}')_{-}, \\
	&\hat{\omega}_{2}(\hat{\omega}_{1}\hat{a}')_{+} - \hat{\omega}_{1}(\hat{\omega}_{2}\hat{a}')_{-} + \hat{\omega}_{1}(\omega_{2}a')_{+} + \hat{\omega}_{2}(\omega_{1}a')_{+} \Bigr),
\end{align*}
which then induces the multiplication on the tangent space via the map $\eta$:
\begin{equation}\label{mul}
	\vec{\xi}_{1} \circ \vec{\xi}_{2} = \eta \left(\eta^{-1}(\vec{\xi}_{1}) \circ \eta^{-1}(\vec{\xi}_{2}) \right),\quad \vec{\xi}_{1},\vec{\xi}_{2}\in T_{\vec{a}}\mathcal{M}.
\end{equation}
This definition is well-defined, as it is independent of the choice of preimages, and possesses an identity element $e$ given by:
\[
e =
\begin{cases}
	\displaystyle \sum_{i=1}^{m}\left(\frac{\partial}{\partial t_{i,0}} + \frac{\partial}{\partial h_{i,0}}\right), & n_{0} = 1; \\[10pt]
	\displaystyle \frac{1}{n_{0}} \frac{\partial}{\partial h_{0,n_{0}-1}}, & n_{0} \geq 2.
\end{cases}
\]
The action of this multiplication is also described by a family of linear maps $C_{\vec{\xi}}: T^{\ast}_{\vec{a}}\mathcal{M} \to T_{\vec{a}}\mathcal{M}$, indexed by a tangent vector $\vec{\xi}=(\xi, \hat{\xi}) \in T_{\vec{a}}\mathcal{M}$. This operator is defined for a covector $\vec{\omega}=(\omega, \hat{\omega})$ by the formula:
\begin{equation}\label{mulope}
	C_{\vec{\xi}} (\vec{\omega}) = \Bigl(a'(\omega \xi + \hat{\omega} \hat{\xi})_{-} - \xi (\omega a' + \hat{\omega} \hat{a}')_{-}, \ -\hat{a}'(\omega \xi + \hat{\omega} \hat{\xi})_{+} + \hat{\xi} (\omega a' + \hat{\omega} \hat{a}')_{+}\Bigr).
\end{equation}
It is related to the tangent space multiplication via the identity:
\begin{equation}\label{muloped}
	C_{\vec{\xi}} (\vec{\omega})=\vec{\xi}\circ\eta(\vec{\omega}).
\end{equation}

\subsection{Principal hierarchy}
The Frobenius manifold $\mathcal{M}$ is naturally endowed with the principal hierarchy, which is a system of commuting Hamiltonian flows defined on the loop space of $\mathcal{M}$, denoted by $\frac{\partial}{\partial T^{u,p}}$ for each coordinate $u \in \mathbf{u}$ and integer $p \ge 0$. For the sake of brevity, we list only two representative sets of these flows:
\begin{align*}
	\frac{\partial\vec{a}}{\partial T^{h_{0,j},p}} &= \frac{\Gamma(1-j/n_{0})}{\Gamma(2+p-j/n_{0})} \biggl( \Bigl\{ (a^{1+p-j/n_{0}})_{\infty,\ge 0}, a \Bigr\}, \Bigl\{ (a^{1+p-j/n_{0}})_{\infty,\ge 0}, \hat{a} \Bigr\} \biggr); \\[10pt] 
	\frac{\partial\vec{a}}{\partial T^{h_{k,r},p}} &= \frac{\Gamma(1-r/n_{k})}{\Gamma(2+p-r/n_{k})} \biggl( \Bigl\{ -( \hat{a}^{1+p-r/n_{k}} )_{\varphi_{k}, \le -1}, a \Bigr\}, \Bigl\{ -( \hat{a}^{1+p-r/n_{k}} )_{\varphi_{k}, \le -1}, \hat{a} \Bigr\} \biggr),
\end{align*}
where $r\neq n_{k}$ and the bracket $\{ \cdot, \cdot \}$ for two functions $f(z)$ and $g(z)$ is defined by
\begin{equation}\label{bracket}
	\{f, g\} := \frac{\partial f}{\partial z} \frac{\partial g}{\partial x} - \frac{\partial g}{\partial z} \frac{\partial f}{\partial x}.
\end{equation}
The remaining flows, particularly those corresponding to the coordinates $t_{i,s}$ and $h_{k,n_{k}}$, are considerably more involved; their full explicit forms can be found in \cite{ma2024infinite}. A key result established in \cite{ma2024infinite} is that the genus-zero Whitham hierarchy is a subhierarchy of this principal hierarchy.

\section{Proofs of the Main Theorems}
In this section, we provide the proofs of our main results, Theorem \ref{mainthm1} and Theorem \ref{mainthm2}. We begin by reviewing the essential concepts of flat F-manifolds and their associated principal hierarchies.

\subsection{Flat F-manifolds and principal hierarchy}
First, we recall the definition of a flat F-manifold, following \cite{manin2005f,arsie2018flat,alcolado2017extended}. A flat F-manifold is a quadruple $(\mathcal{N},\nabla,\circ,e)$ consisting of an analytic manifold $\mathcal{N}$, a flat and torsion-free connection $\nabla$ on its tangent bundle, a commutative and associative multiplication $\circ$ on each tangent space, and a unity vector field $e$ satisfying $\nabla e=0$. Furthermore, the multiplication must be expressible in terms of a vector potential $\Psi$ as
$$X\circ Y=[X,[Y,\Psi]]$$
for any flat vector fields $X$ and $Y$.

A primary source of such structures is the set of open WDVV equations:
$$c_{\alpha\beta}^{\delta}\frac{\partial F^{o}(\mathbf{t},s)}{\partial t^{\delta}\partial t^{\gamma}}+\frac{\partial F^{o}(\mathbf{t},s)}{\partial t^{\alpha}\partial t^{\beta}}\frac{\partial F^{o}(\mathbf{t},s)}{\partial t^{\gamma}\partial s}=c_{\beta\gamma}^{\delta}\frac{\partial F^{o}(\mathbf{t},s)}{\partial t^{\delta}\partial t^{\alpha}}+\frac{\partial F^{o}(\mathbf{t},s)}{\partial t^{\beta} \partial t^{\gamma}}\frac{\partial F^{o}(\mathbf{t},s)}{\partial t^{\alpha}\partial s}$$
and
$$c_{\alpha\beta}^{\delta}\frac{\partial F^{o}(\mathbf{t},s)}{\partial t^{\delta} \partial s}+\frac{\partial F^{o}(\mathbf{t},s)}{\partial t^{\alpha} \partial t^{\beta}}\frac{\partial F^{o}(\mathbf{t},s)}{\partial s\partial s}=\frac{\partial F^{o}(\mathbf{t},s)}{\partial t^{\alpha} \partial s}\frac{\partial F^{o}(\mathbf{t},s)}{\partial t^{\beta}\partial s},$$
where $c_{\alpha\beta}^{\delta}$ are the structure constants of an underlying Frobenius manifold $N$ (we assume the Einstein summation convention for repeated indices). As observed by P. Rossi, any solution $F^{o}(\mathbf{t},s)$ to these equations defines a flat F-manifold structure on the product space $N\times U$, where $U$ is a certain domain in $\mathbb{C}$ with coordinate $s$. The vector potential for this structure is given by
\begin{equation}\label{vecpot2}
	\Psi=\frac{\partial F(\mathbf{t})}{\partial t^{\alpha}}\eta^{\alpha\beta}\frac{\partial}{\partial t^{\beta}}+F^{o}(\mathbf{t},s)\frac{\partial}{\partial s}.
\end{equation}
Here, $F(\mathbf{t})$ is the prepotential of $N$.

The notion of a principal hierarchy extends to this setting\cite{arsie2018flat}. A calibration for a flat F-manifold is a set of vector fields $\{\Theta_{t^{\alpha},p}\}_{p\ge 0}$ satisfying the recursive relation
\begin{equation}\label{recur}
	\nabla_{X}\Theta_{t^{\alpha},p+1}= \Theta_{t^{\alpha},p}\circ X,
\end{equation}
for any vector field $X$, with the initial condition $\Theta_{t^{\alpha},0}=\frac{\partial}{\partial t^{\alpha}}$, where $\mathbf{t}=\{t^\alpha\}$ are the flat coordinates of the connection $\nabla$. The associated principal hierarchy is then the system of flows on the loop space of $\mathcal{N}$ defined by:
\begin{equation}\label{fpridef}
	\frac{\partial t^{\alpha}(x)}{\partial T^{t^{\beta},p}}=\Theta_{t^{\beta},p;\delta}(\mathbf{t}(x))c_{\delta\gamma}^{\alpha} \frac{\partial t^{\gamma}(x)}{\partial x},
\end{equation}
where $\Theta_{t^{\beta},p;\delta}$ denotes the $\frac{\partial}{\partial t^{\delta}}$-component of $\Theta_{t^{\beta},p}$.

\subsection{Proof of Theorem \ref{mainthm1}}
The proof of Theorem \ref{mainthm1} relies on a key result from \cite{alcolado2017extended}, which provides a method to construct solutions to the open WDVV equations.

\begin{lem}\label{mainlem}
	Let $N$ be a Frobenius manifold with flat coordinates $\mathbf{t}=\{t^{1}, \dots, t^{n}\}$, and let $f(\mathbf{t},s)$ be a smooth function on $N\times U$ such that $\partial_{s} f\ne 0$. If $f$ satisfies the equation
	\begin{equation}\label{omegaequ}
		\partial_{\alpha}\partial_{\beta}f = \partial_{s}\left(\frac{\partial_{\alpha}f \partial_{\beta}f - c_{\alpha\beta}^{\delta}\partial_{\delta}f}{\partial_{s}f}\right), \quad \alpha,\beta=1,\dots,n,
	\end{equation}
	where $\partial_{\alpha}=\frac{\partial}{\partial t^{\alpha}}$, then the function $F^{o}=F^{o}(\mathbf{t},s)$ determined by the system
	\begin{equation}\label{Fosys}
		\partial_{\alpha}\partial_{\beta} F^{o} = \frac{\partial_{\alpha}f \partial_{\beta}f - c_{\alpha\beta}^{\delta}\partial_{\delta}f}{\partial_{s}f}, \quad \partial_{s}F^o=f,
	\end{equation}
	provides a solution to the open WDVV equations for $N$.
\end{lem}

Our proof strategy is to apply Lemma \ref{mainlem} to the infinite-dimensional Frobenius manifold $\mathcal{M}$. We will show that the functions $f=a(s)$ and $f=\hat{a}(s)$ satisfy the required condition \eqref{omegaequ}. To this end, we first need an explicit expression for the right-hand side of \eqref{omegaequ} in our specific setting, which is derived in the following lemma.

\begin{lem}\label{lem:identities}
	For any tangent vectors $\partial_{1}, \partial_{2} \in T_{\vec{a}}\mathcal{M}$, the following identities hold:
	\begin{align}
		\begin{split} \label{mullem1}
			\frac{\partial_{1}a(z)\partial_{2}a(z)-(\partial_{1}\circ\partial_{2})a(z)}{a'(z)} 
			&= \biggl(\frac{\partial_1 \zeta(z) \partial_2 \zeta(z)}{\zeta'(z)}\biggr)_{\!-} 
			+ \biggl(\frac{\partial_1 \ell(z) \partial_2 \ell(z)}{\ell'(z)}\biggr)_{\!\infty,\,\ge 0} \\
			&\quad + \sum_{j=1}^{m}\biggl(\frac{\partial_1 \ell(z) \partial_2 \ell(z)}{\ell'(z)}\biggr)_{\!\varphi_{j},\,\le -1},
		\end{split} \\[15pt]
		\begin{split} \label{mullem2}
			\frac{\partial_{1}\hat{a}(z)\partial_{2}\hat{a}(z)-(\partial_{1}\circ\partial_{2})\hat{a}(z)}{\hat{a}'(z)} 
			&= -\biggl(\frac{\partial_1 \zeta(z) \partial_2 \zeta(z)}{\zeta'(z)}\biggr)_{\!+} 
			+ \biggl(\frac{\partial_1 \ell(z) \partial_2 \ell(z)}{\ell'(z)}\biggr)_{\!\infty,\,\ge 0} \\
			&\quad + \sum_{j=1}^{m}\biggl(\frac{\partial_1 \ell(z) \partial_2 \ell(z)}{\ell'(z)}\biggr)_{\!\varphi_{j},\,\le -1}.
		\end{split}
	\end{align}
\end{lem}
 \begin{proof}
We prove only \eqref{mullem1}; the proof of \eqref{mullem2} is entirely analogous.

Denote $\vec{\xi}_{\nu}=(\xi_{\nu}(z),\hat{\xi}_{\nu}(z))=\partial_{\nu}\vec{a}$ for $\nu=1,2$. Since the map $\eta$ is surjective, we can write $\vec{\xi}_{\nu}=\eta(\vec{\omega}_{\nu})$ for some $\vec{\omega}_{\nu}=(\omega_{\nu}(z),\hat{\omega}_{\nu}(z))\in\mathcal{H}\times\mathcal{H}$. From the definition of the tangent space multiplication \eqref{mul} and the operator $C_{\vec{\xi}}$ \eqref{muloped}, we have
\[
(\partial_{1}\circ\partial_{2})\vec{a} = \vec{\xi}_{1}\circ\vec{\xi}_{2} = C_{\vec{\xi}_{1}}\vec{\omega}_{2}.
\]
By substituting the explicit formulas for $\xi_{2}(z)$ (as the first component of $\eta(\vec{\omega}_2)$) and the first component of $C_{\vec{\xi}_{1}}\vec{\omega}_{2}$, a direct calculation shows that
\begin{align*}
	\partial_{1}a(z)\partial_{2}a(z) &= \xi_{1}(z)\Bigl(a'(z)(\omega_{2}(z)+\hat{\omega}_{2}(z))_{-} - (\omega_{2}(z) a'(z)+\hat{\omega}_{2}(z)\hat{a}'(z))_{-}\Bigr), \\
	(\partial_{1}\circ\partial_{2})a(z) &= a'(z)(\omega_{2}(z)\xi_{1}(z)+\hat{\omega}_{2}(z)\hat{\xi}_{1}(z))_{-} - \xi_{1}(z)(\omega_{2}(z)a'(z)+\hat{\omega}_{2}(z)\hat{a}'(z))_{-}.
\end{align*}
Thus, dividing the difference by $a'(z)$, the left-hand side of \eqref{mullem1} simplifies to:
\begin{equation*}
	\frac{\partial_{1}a(z)\partial_{2}a(z)-(\partial_{1}\circ\partial_{2})a(z)}{a'(z)} = \xi_{1}(z)(\omega_{2}(z)+\hat{\omega}_{2}(z))_{-} - (\xi_{1}(z)\omega_{2}(z)+\hat{\xi}_{1}(z)\hat{\omega}_{2}(z))_{-}.
\end{equation*}
 	We now express the terms on the right-hand side of \eqref{mullem1} in terms of $\vec{\xi}_{1}$ and $\vec{\omega}_{2}$. This is accomplished using the following identities from \cite{ma2024infinite}:
 	\begin{align}
 		\xi_{\nu}(z) - \hat{\xi}_{\nu}(z) &= -\zeta'(z)(\omega_{\nu}(z)_{+} - \hat{\omega}_{\nu}(z)_{-}), \label{etalem1} \\[1pt]
 		\biggl(\frac{\xi_{\nu}(z)}{a'(z)}\biggr)_{\!\infty,\,\ge -n_{0}+1} &= \biggl( (\omega_{\nu}(z) + \hat{\omega}_{\nu}(z))_{-} \biggr)_{\!\infty,\,\ge -n_{0}+1}, \label{etalem2} \\[1pt]
 		\biggl(\frac{\hat{\xi}_{\nu}(z)}{\hat{a}'(z)}\biggr)_{\!\varphi_{j},\,\le n_{j}} &= -\biggl( (\omega_{\nu}(z) + \hat{\omega}_{\nu}(z))_{+} \biggr)_{\!\varphi_{j},\,\le n_{j}}, \quad 1\le j\le m. \label{etalem3}
 	\end{align}
 	Applying these identities, we obtain the terms on the right-hand side of \eqref{mullem1} as:
 	\begin{align*}
 		\frac{\partial_1 \zeta(z) \partial_2 \zeta(z)}{\zeta'(z)} 
 		&= -(\xi_{1}(z)-\hat{\xi}_{1}(z))(\omega_{2}(z)_{+}-\hat{\omega}_{2}(z)_{-}), \\[1pt]	
 		\biggl(\frac{\partial_1 \ell(z) \partial_2 \ell(z)}{\ell'(z)}\biggr)_{\!\infty,\,\ge 0} 
 		&= \biggl(\frac{\xi_{1}(z) \xi_{2}(z)}{a'(z)}\biggr)_{\!\infty,\,\ge 0} 
 		= \bigl( (\omega_{2}(z)+\hat{\omega}_{2}(z))_{-}\xi_{1}(z) \bigr)_{+}, \\[1pt]
 		\sum_{j=1}^{m}\biggl(\frac{\partial_1 \ell(z) \partial_2 \ell(z)}{\ell'(z)}\biggr)_{\!\varphi_{j},\,\le -1} 
 		&= -\bigl( (\omega_{2}(z)+\hat{\omega}_{2}(z))_{+}\hat{\xi}_{1}(z) \bigr)_{-}.
 	\end{align*}
 	Summing these three terms recovers the expression for the left-hand side. This completes the proof of \eqref{mullem1}.
 \end{proof}
 
We now complete the proof of Theorem \ref{mainthm1}. Since the coordinates $\mathbf{u}$ are flat, the covariant derivatives of the basis vector fields vanish, i.e., $\nabla_{\partial_{\alpha}}\partial_{\beta}=0$ for any $\alpha,\beta\in \mathbf{u}$. Substituting this into the definition of the Levi-Civita connection \eqref{conn}, we find that
\begin{align*}
	\partial_{\alpha}\partial_{\beta}a(z) &= \biggl(\frac{\partial_\alpha \zeta(z) \partial_\beta \zeta(z)}{\zeta'(z)}\biggr)_{\!-}' + \biggl(\frac{\partial_\alpha \ell(z) \partial_\beta \ell(z)}{\ell'(z)}\biggr)_{\!\infty,\,\ge 0}' + \sum_{j=1}^{m}\biggl(\frac{\partial_\alpha \ell(z) \partial_\beta \ell(z)}{\ell'(z)}\biggr)_{\!\varphi_{j},\,\le -1}' \\
	\intertext{and}
	\partial_{\alpha}\partial_{\beta}\hat{a}(z) &= -\biggl(\frac{\partial_\alpha \zeta(z) \partial_\beta \zeta(z)}{\zeta'(z)}\biggr)_{\!+}' + \biggl(\frac{\partial_\alpha \ell(z) \partial_\beta \ell(z)}{\ell'(z)}\biggr)_{\!\infty,\,\ge 0}' + \sum_{j=1}^{m}\biggl(\frac{\partial_\alpha \ell(z) \partial_\beta \ell(z)}{\ell'(z)}\biggr)_{\!\varphi_{j},\,\le -1}'.
\end{align*}
A comparison of the above expressions for $\partial_\alpha \partial_\beta a(z)$ and $\partial_\alpha \partial_\beta \hat{a}(z)$ with the identities in Lemma \ref{lem:identities} shows that the condition \eqref{omegaequ} of Lemma \ref{mainlem} is satisfied for $f=a(s)$ and $f=\hat{a}(s)$, respectively. The existence of the solutions $\Omega$ and $\hat{\Omega}$ as defined in Theorem \ref{mainthm1} is thus established. \hfill \qed

\subsection{Proof of Theorem \ref{mainthm2}}
We now construct the principal hierarchy for the flat F-manifold structures associated with the solutions $\Omega$ and $\hat{\Omega}$. By definition, this requires us to determine a set of vector fields
\begin{equation*}
	\Theta_{u,p} = \sum_{\alpha\in\mathbf{u}} \Theta_{u,p;\alpha} \partial_{\alpha} + \Theta_{u,p;s} \partial_{s}, \quad u\in \mathbf{u}\cup\{s\}, \quad p\ge 0,
\end{equation*}
on the extended manifold $\mathcal{M}\times U$. These vector fields are defined recursively by the relation
\begin{equation}\label{openre1}
	\nabla_{X}^{\mathrm{ext}}\Theta_{u,p+1} = X \star \Theta_{u,p}
\end{equation}
for any vector field $X=\sum_{u\in \mathbf{u}\cup\{s\}}X_{u}\partial_{u}$ on $\mathcal{M}\times U$, together with the initial condition
\begin{equation}\label{openre2}
	\Theta_{u,0} = \partial_{u}, \quad u\in \mathbf{u}\cup\{s\}.
\end{equation}
Here, $\nabla^{\mathrm{ext}}$ is the affine connection on $\mathcal{M}\times U$ with respect to the flat coordinates $\mathbf{u}\cup\{s\}$. The operation $\star$ is the multiplication of the flat F-manifold structure, with its components given by:
\begin{align*}
	\partial_{\alpha} \star \partial_{\beta} &= \partial_{\alpha} \circ \partial_{\beta} + \partial_{\alpha}\partial_{\beta}F^{o}(\mathbf{u},s)\partial_{s}, \quad \alpha,\beta\in \mathbf{u}; \\[1pt]
	\partial_{\alpha} \star \partial_{s} &= \partial_{\alpha}f(s)\partial_{s}; \\[1pt]
	\partial_{s} \star \partial_{s} &= f'(s)\partial_{s}.
\end{align*}
This general framework applies to our two cases, corresponding to $(F^o, f) = (\Omega, a(s))$ and $(F^o, f) = (\hat{\Omega}, \hat{a}(s))$, respectively.

For the case of $f=a(s)$ and $F^{o}=\Omega$, we propose the following ansatz for the vector fields $\Theta_{u,p}$:
\begin{align}
	\sum_{\alpha\in\mathbf{u}}\Theta_{u,p;\alpha}\partial_{\alpha}\vec{a}(z) &= \eta\biggl( \frac{\partial Q_{u,p}(a(z),\hat{a}(z))}{\partial a}, \frac{\partial Q_{u,p}(a(z),\hat{a}(z))}{\partial \hat{a}} \biggr), \label{privec1} \\[5pt]
	\Theta_{u,p;s} &= -(Q_{u,p-1}(a(s),\hat{a}(s)))_{-} + \widetilde{Q}_{u,p-1}(a(s)) \label{privec2}
\end{align}
for indices $u\in\mathbf{u}$, and
\begin{equation}\label{prisvec}
	\Theta_{s,p;\alpha}=0, \quad \alpha\in\mathbf{u}; \qquad \Theta_{s,p;s}=\frac{a(s)^{p}}{p!},
\end{equation}
where $Q_{u,p}$ and $\widetilde{Q}_{u,p}$ are required to satisfy
\begin{equation}\label{Qrec}
	\frac{\partial Q_{u,p}}{\partial a} + \frac{\partial Q_{u,p}}{\partial \hat{a}} = Q_{u,p-1}, \quad \frac{d \widetilde{Q}_{u,p}}{d a} = \widetilde{Q}_{u,p-1}, \quad p\ge 0.
\end{equation}
 
We now verify that this ansatz satisfies the recursive relation \eqref{openre1}. Expanding the left-hand side of \eqref{openre1}, we obtain:\footnote{We remark that although the sums in the following expansions range over the infinite set of indices $\mathbf{u}$, they represent well-defined vector fields on $\mathcal{M}$. This is guaranteed by the fact that the connection $\nabla$ and the multiplication $\circ$ are intrinsically defined via the explicit formulas given in Section 2 (see \eqref{conn} and \eqref{mul}, respectively).}
\begin{align*}
	\nabla_{X}^{\mathrm{ext}}\Theta_{u,p+1} &= \sum_{\alpha\in\mathbf{u}}\nabla_{X_{\alpha}\partial_{\alpha}}\Theta_{u,p+1} + X_{s}\nabla_{\partial_{s}}\Theta_{u,p+1} \\
	&= \sum_{\alpha\in\mathbf{u}}\nabla_{X_{\alpha}\partial_{\alpha}}\biggl(\sum_{\beta\in\mathbf{u}}\Theta_{u,p+1;\beta}\partial_{\beta}\biggr) + \sum_{\alpha\in\mathbf{u}}X_{\alpha}(\partial_{\alpha}\Theta_{u,p+1;s})\partial_{s} \\
	&\quad + \sum_{\alpha\in\mathbf{u}}X_{s}(\partial_{s}\Theta_{u,p+1;\alpha})\partial_{\alpha} + X_{s}(\partial_{s}\Theta_{u,p+1;s})\partial_{s}.
\end{align*}

Similarly, expanding the $\star$-product on the right-hand side of \eqref{openre1} yields:
\begin{align*}
	X \star \Theta_{u,p} &= \sum_{\alpha,\beta\in \mathbf{u}}X_{\alpha}\Theta_{u,p;\beta}\partial_{\alpha}\circ\partial_{\beta} + \sum_{\alpha,\beta\in \mathbf{u}}X_{\alpha}\Theta_{u,p;\beta}(\partial_{\alpha}\partial_{\beta}\Omega)\partial_{s} \\
	&\quad + \sum_{\alpha\in\mathbf{u}}X_{\alpha}\Theta_{u,p;s}(\partial_{\alpha}\partial_{s}\Omega)\partial_{s} + \sum_{\alpha\in\mathbf{u}}X_{s}\Theta_{u,p;\alpha}(\partial_{\alpha}\partial_{s}\Omega)\partial_{s} \\
	&\quad + X_{s}\Theta_{u,p;s}(\partial_{s}^{2}\Omega)\partial_{s}.
\end{align*}
The following identity has been established in the construction of the principal hierarchy for $\mathcal{M}$ in \cite{ma2024infinite}:
\[
\sum_{\alpha\in\mathbf{u}}\nabla_{X_{\alpha}\partial_{\alpha}}\biggl(\sum_{\beta\in\mathbf{u}}\Theta_{u,p+1;\beta}\partial_{\beta}\biggr) = \sum_{\alpha,\beta\in \mathbf{u}}X_{\alpha}\Theta_{u,p;\beta}\partial_{\alpha}\circ\partial_{\beta}.
\]
Furthermore, since the components $\Theta_{u,p;\alpha}$ are independent of $s$, we have
\[
\partial_{s}\Theta_{u,p;\alpha}=0, \quad \alpha\in\mathbf{u}.
\]
The verification of \eqref{openre1} thus reduces to proving the equality of the components along the vector field $\partial_s$, as follows:
\begin{align*}
	\sum_{\alpha\in\mathbf{u}}X_{\alpha}\partial_{\alpha}\Theta_{u,p+1;s} + X_{s}\partial_{s}\Theta_{u,p+1;s} &= \sum_{\alpha,\beta\in \mathbf{u}}X_{\alpha}\Theta_{u,p;\beta}\partial_{\alpha}\partial_{\beta}\Omega + \sum_{\alpha\in\mathbf{u}}X_{\alpha}\Theta_{u,p;s}(\partial_{\alpha}\partial_{s}\Omega) \\
	&\quad + \sum_{\alpha\in\mathbf{u}}X_{s}\Theta_{u,p;\alpha}(\partial_{\alpha}\partial_{s}\Omega) + X_{s}\Theta_{u,p;s}\partial_{s}^{2}\Omega.
\end{align*}
Using the explicit form of $\Theta_{u,p+1;s}$ from our ansatz, the left-hand side becomes:
\begin{align*}
	& -\biggl(\sum_{\alpha\in\mathbf{u}}X_{\alpha}\partial_{\alpha}Q_{u,p}(\vec{a}(s))\biggr)_{\!-} - X_{s}(\partial_{s}Q_{u,p}(\vec{a}(s)))_{-} \\
	&\qquad + \sum_{\alpha\in\mathbf{u}}X_{\alpha}\partial_{\alpha}\widetilde{Q}_{u,p}(a(s)) + X_{s}\partial_{s}\widetilde{Q}_{u,p}(a(s)).
\end{align*}
We analyze the three groups of terms on the right-hand side separately. The first term, involving $\partial_{\alpha}\partial_{\beta}\Omega$, can be expressed using Theorem \ref{mainthm1} and Lemma \ref{lem:identities} as:
\begin{align*}
	&\frac{\sum_{\alpha,\beta\in\mathbf{u}}X_{\alpha}\Theta_{u,p;\beta}\partial_{\alpha}a(s)\partial_{\beta}a(s) - \bigl(\partial_{X}\circ\partial_{\sum_{\alpha\in\mathbf{u}}\Theta_{u,p;\alpha}\partial_{\alpha}}\bigr)a(s)}{a'(s)} \\
	&\quad = -\frac{(\sum_{\alpha\in\mathbf{u}}X_{\alpha}\partial_{\alpha}a(s))(Q_{u,p}(\vec{a}(s)))_{\!-}'}{a'(s)} + \biggl(\sum_{\alpha\in\mathbf{u}}X_{\alpha}\partial_{\alpha}a(s)\biggr)(Q_{u,p-1}(\vec{a}(s)))_{-} \\
	&\qquad + \frac{(\sum_{\alpha\in\mathbf{u}}X_{\alpha}\partial_{\alpha}a(s))(Q_{u,p}(\vec{a}(s)))_{\!-}'}{a'(s)} - \biggl(\sum_{\alpha\in\mathbf{u}}X_{\alpha}\partial_{\alpha}Q_{u,p}(\vec{a}(s))\biggr)_{\!-} \\
	&\quad = \biggl(\sum_{\alpha\in\mathbf{u}}X_{\alpha}\partial_{\alpha}a(s)\biggr)(Q_{u,p-1}(\vec{a}(s)))_{-} - \biggl(\sum_{\alpha\in\mathbf{u}}X_{\alpha}\partial_{\alpha}Q_{u,p}(\vec{a}(s))\biggr)_{\!-}.
\end{align*}
The second term, involving $\partial_\alpha \partial_s \Omega$, evaluates to:
\begin{equation}\label{mulright1}
	\sum_{\alpha\in\mathbf{u}}X_{\alpha}\Theta_{u,p;s}(\partial_{\alpha}\partial_{s}\Omega) = -\biggl(\sum_{\alpha\in\mathbf{u}}X_{\alpha}\partial_{\alpha}a(s)\biggr)(Q_{u,p-1}(\vec{a}(s)))_{-} + \biggl(\sum_{\alpha\in\mathbf{u}}X_{\alpha}\partial_{\alpha}a(s)\biggr)\widetilde{Q}_{u,p-1}(a(s)).
\end{equation}
The last two terms simplify to:
\begin{align}
	&\sum_{\alpha\in\mathbf{u}}X_{s}\Theta_{u,p;\alpha}(\partial_{\alpha}\partial_{s}\Omega) + X_{s}\Theta_{u,p;s}\partial_{s}^{2}\Omega \nonumber \\
	&\quad = X_{s}\Bigl(-(Q_{u,p}(\vec{a}(s)))_{\!-}' + a'(s)(Q_{u,p-1}(\vec{a}(s)))_{-}\Bigr) \nonumber \\
	&\qquad - X_{s}a'(s)(Q_{u,p-1}(\vec{a}(s)))_{-} + X_{s}a'(s)\widetilde{Q}_{u,p-1}(a(s)) \nonumber \\
	&\quad = -X_{s}(Q_{u,p}(\vec{a}(s)))_{\!-}' + X_{s}a'(s)\widetilde{Q}_{u,p-1}(a(s)). \label{mulright2}
\end{align}
Summing these three results and using the equality \eqref{Qrec}, we find that the right-hand side is identical to the expression for the left-hand side. This confirms the identity \eqref{openre1}.

By choosing the explicit forms for $Q_{u,p}$ and $\widetilde{Q}_{u,p}$ as follows:\footnote{Strictly speaking, $Q_{u,p}$ involves fractional powers and logarithms, requiring restrictions on winding numbers. However, following the strategy in \cite{ma2024infinite}, one can derive explicit expressions for $\Theta_{u,p}$ that are independent of these branch cuts. Thus, the constructed vector fields are globally well-defined.}
\begin{align*}
	Q_{t_{i,l},p} &= \frac{1}{(p+1)!} \frac{d_{i}}{l+d_{i}} \zeta(z)^{l/d_{i}} \bigl(a(z)^{p+1} -\hat{a}(z)^{p+1}\bigr)\mathbf{1}_{\gamma_{i}}, \quad l\ne -d_{i}; \\[10pt]
	Q_{h_{0,j},p} &= \frac{\Gamma(1-j/n_{0})}{\Gamma(2+p-j/n_{0})} a(z)^{1+p-j/n_{0}}; \\[10pt]
	Q_{h_{k,r},p} &= \frac{\Gamma(1-r/n_{k})}{\Gamma(2+p-r/n_{k})} \hat{a}(z)^{1+p-r/n_{k}}\mathbf{1}_{\gamma_{k}}, \quad r\ne n_{k}; \\[10pt]
	Q_{t_{i,-d_{i}},p} &= -d_{i}\frac{a(z)^{p}}{p!} \biggl(\log \frac{a(z)^{1/n_{0}}}{ \zeta(z)^{1/d_{i}} } -\frac{c_{p}}{n_{0}}\biggr)\mathbf{1}_{\gamma_{i}} \nonumber \\
	&\quad - d_{i}\sum_{j\ne i}\frac{a(z)^{p}}{p!} \biggl(\log a(z)^{1/n_{0}}-\frac{c_{p}}{n_{0}}\biggr)\mathbf{1}_{\gamma_{j}}; \\[10pt]
	Q_{h_{k,n_{k}},p} &= n_{k} \frac{\hat{a}(z)^{p}}{p!} \biggl(\log\bigl(\zeta(z)^{1/d_{k}}\hat{a}(z)^{1/n_{k}}\bigr)-\frac{c_{p}}{n_{k}}\biggr)\mathbf{1}_{\gamma_{k}} - \frac{n_{k}}{d_{k}}Q_{t_{k,-d_{k}},p},
\end{align*}
and
\begin{align*}
	\widetilde{Q}_{t_{i,l},p} &= 0, \quad l\ne -d_{i}; \\[5pt]
	\widetilde{Q}_{h_{0,j},p} &= Q_{h_{0,j},p}; \\[5pt]
	\widetilde{Q}_{h_{k,r},p} &= 0, \quad r\ne n_{k}; \\[5pt]
	\widetilde{Q}_{t_{i,-d_{i}},p} &= -d_{i}\frac{a(z)^{p}}{p!}\biggl(\log a(z)^{1/n_{0}}-\frac{c_{p}}{n_{0}}\biggr); \\[5pt]
	\widetilde{Q}_{h_{k,n_{k}},p} &= n_{k}\frac{a(z)^{p}}{p!}\biggl(\log a(z)^{1/n_{0}}-\frac{c_{p}}{n_{0}}\biggr),
\end{align*}
with
\begin{equation*}
	Q_{u,-1} := \frac{\partial Q_{u,0}}{\partial a} + \frac{\partial Q_{u,0}}{\partial \hat{a}}, \quad \widetilde{Q}_{u,-1} := \frac{d \widetilde{Q}_{u,0}}{d a},
\end{equation*}
one can then verify that this choice ensures that the resulting vector fields $\Theta_{u,p}$ satisfy the initial condition \eqref{openre2}.
 	
 For the case of $f=\hat{a}(s)$ and $F^{o}=\hat{\Omega}$, the ansatz for the vector fields $\Theta_{u,p}$ is given by:
 \begin{align}
 	\sum_{\alpha\in\mathbf{u}}\Theta_{u,p;\alpha}\partial_{\alpha}\vec{a}(z) &= \eta\biggl( \frac{\partial Q_{u,p}(a(z),\hat{a}(z))}{\partial a}, \frac{\partial Q_{u,p}(a(z),\hat{a}(z))}{\partial \hat{a}} \biggr), \\
 	\Theta_{u,p;s} &= (Q_{u,p-1}(a(s),\hat{a}(s)))_{+} + \widetilde{Q}_{u,p-1}(\hat{a}(s)) \label{privec3}
 \end{align}
 for $u\in\mathbf{u}$, and
 \begin{equation}
 	\Theta_{s,p;\alpha}=0, \quad \alpha\in\mathbf{u}; \qquad \Theta_{s,p;s}=\frac{\hat{a}(s)^{p}}{p!},
 \end{equation}
 where $Q_{u,p}(a,\hat{a})$ and $\widetilde{Q}_{u,p}(\hat{a})$ are required to satisfy
 \begin{equation*}
 	\frac{\partial Q_{u,p}}{\partial a} + \frac{\partial Q_{u,p}}{\partial \hat{a}} = Q_{u,p-1}, \quad \frac{d \widetilde{Q}_{u,p}}{d \hat{a}} = \widetilde{Q}_{u,p-1}, \quad p\ge 0.
 \end{equation*}
 A similar calculation shows that the recursive relation \eqref{openre1} holds. In this case, the functions $Q_{u,p}$ are taken to be the same as in the previous case, while $\widetilde{Q}_{u,p}(\hat{a})$ are given by:
 \begin{align*}
 	\widetilde{Q}_{t_{i,l},p} &= 0, \quad l\ne -d_{i}; \\[5pt]
 	\widetilde{Q}_{h_{0,j},p} &= 0; \\[5pt]
 	\widetilde{Q}_{h_{k,r},p} &= -Q_{h_{k,r},p}, \quad r\ne n_{k}; \\[5pt]
 	\widetilde{Q}_{t_{i,-d_{i}},p} &= 0; \\[5pt]
 	\widetilde{Q}_{h_{k,n_{k}},p} &= -n_{k}\frac{\hat{a}(z)^{p}}{p!}\biggl(\log \bigl(\hat{a}(z)^{1/n_{k}}\bigr) - \frac{c_{p}}{n_{k}}\biggr)\mathbf{1}_{\gamma_{k}},
 \end{align*}
 with
 \begin{equation*}
 	\widetilde{Q}_{u,-1} := \frac{d \widetilde{Q}_{u,0}}{d \hat{a}}.
 \end{equation*}
 It can then be verified that this construction satisfies the initial condition \eqref{openre2}.

We now use the calibration $\Theta_{u,p}$ defined above to construct the corresponding principal hierarchy. The flows of this hierarchy, denoted by $\frac{\partial}{\partial\widetilde{T}^{u,p}}$, are given by the definition \eqref{fpridef}:
\begin{align*}
	\frac{\partial}{\partial\widetilde{T}^{u,p}} &= \Theta_{u,p} \star \sum_{v\in\mathbf{u}\cup\{s\}}v_{x}\partial_{v} \\
	&= \sum_{\alpha\in \mathbf{u}}\Theta_{u,p;\alpha}\partial_{\alpha} \circ \sum_{v\in\mathbf{u}\cup\{s\}}v_{x}\partial_{v} + \sum_{v,\beta\in \mathbf{u}}v_{x}\Theta_{u,p;\beta}(\partial_{v}\partial_{\beta}F^{o})\partial_{s} \\
	&\quad + \sum_{v\in\mathbf{u}}v_{x}\Theta_{u,p;s}(\partial_{v}\partial_{s}F^{o})\partial_{s} + \sum_{\alpha\in\mathbf{u}}s_{x}\Theta_{u,p;\alpha}(\partial_{\alpha}\partial_{s}F^{o})\partial_{s} + s_{x}\Theta_{u,p;s}(\partial_{s}^{2}F^{o})\partial_{s}.
\end{align*}
We denote the principal hierarchy for the underlying Frobenius manifold $\mathcal{M}$ by
\begin{equation*}
	\frac{\partial}{\partial T^{u,p}} = \sum_{\alpha\in \mathbf{u}}\Theta_{u,p;\alpha}\partial_{\alpha} \circ \sum_{v\in\mathbf{u}\cup\{s\}}v_{x}\partial_{v}, \quad u\in \mathbf{u}.
\end{equation*}
For the case of $(F^o, f) = (\Omega, a(s))$, an application of the identities \eqref{mulright1} and \eqref{mulright2} yields the flows for $u\in\mathbf{u}$ as:
\begin{align*}
	\frac{\partial}{\partial\widetilde{T}^{u,p}} &= \frac{\partial}{\partial T^{u,p}} - \biggl( \sum_{v\in\mathbf{u}}v_{x}\partial_{v}Q_{u,p}(\vec{a}(s)) + s_{x}\partial_{s}Q_{u,p}(\vec{a}(s)) \biggr)_{\!-} \partial_{s} \\
	&\qquad + \biggl( \sum_{v\in\mathbf{u}}v_{x}\partial_{v}\widetilde{Q}_{u,p}(a(s)) + s_{x}\partial_{s}\widetilde{Q}_{u,p}(a(s)) \biggr) \partial_{s} \\
	&= \frac{\partial}{\partial T^{u,p}} - \partial_{x}(Q_{u,p}(\vec{a}(s)))_{-}\partial_{s} + \partial_{x}\widetilde{Q}_{u,p}(a(s))\partial_{s},
\end{align*}
and
\begin{equation*}
	\frac{\partial}{\partial\widetilde{T}^{s,p-1}} = \partial_{x}\biggl(\frac{a(s)^{p}}{p!}\biggr)\partial_{s}.
\end{equation*}
For the case of $(F^o, f) = (\hat{\Omega}, \hat{a}(s))$, a similar calculation shows that
\begin{align*}
	\frac{\partial}{\partial\widetilde{T}^{u,p}} &= \frac{\partial}{\partial T^{u,p}} + \biggl( \sum_{v\in\mathbf{u}}v_{x}\partial_{v}Q_{u,p}(\vec{a}(s)) + s_{x}\partial_{s}Q_{u,p}(\vec{a}(s)) \biggr)_{\!+} \partial_{s} \\
	&\qquad + \biggl( \sum_{v\in\mathbf{u}}v_{x}\partial_{v}\widetilde{Q}_{u,p}(\hat{a}(s)) + s_{x}\partial_{s}\widetilde{Q}_{u,p}(\hat{a}(s)) \biggr) \partial_{s} \\
	&= \frac{\partial}{\partial T^{u,p}} + \partial_{x}(Q_{u,p}(\vec{a}(s)))_{+}\partial_{s} + \partial_{x}\widetilde{Q}_{u,p}(\hat{a}(s))\partial_{s}, \quad u\in\mathbf{u},
\end{align*}
and
\begin{equation*}
	\frac{\partial}{\partial\widetilde{T}^{s,p-1}} = \partial_{x}\biggl(\frac{\hat{a}(s)^{p}}{p!}\biggr)\partial_{s}.
\end{equation*}
Substituting the explicit forms of $Q_{u,p}$ and $\widetilde{Q}_{u,p}$ into these expressions yields the explicit form of the principal hierarchy, thus completing the proof of Theorem \ref{mainthm2}. \hfill \qed

 \section{Reductions}
 In this section, we prove Corollaries \ref{maincor4}--\ref{maincor6} by investigating the reductions of the construction presented in the previous section to specific submanifolds of $\mathcal{M}$.
 
 \subsection{Proof of Corollary \ref{maincor4}}
 Recall that the Frobenius manifold $M$ consists of rational functions of the form
 \begin{equation}\label{ell}
 	\ell(z) = z^{n_0} + a_{n_0-2}z^{n_0-2} + \cdots + a_1 z + a_{0} + \sum_{i=1}^{m}\sum_{j=1}^{n_i}b_{i,j}(z-\varphi_{i})^{-j},
 \end{equation}
 equipped with the flat metric
 \[
 \langle \partial', \partial'' \rangle_{\eta} = \sum_{|\ell|<\infty} \operatorname*{Res}_{d\ell=0} \frac{\partial'(\ell(z)dz) \, \partial''(\ell(z)dz)}{d\ell(z)},
 \]
 and a multiplication $\circ$ defined by the symmetric 3-tensor
 \[
 c(\partial', \partial'', \partial''') := \sum_{|\ell|<\infty} \operatorname*{Res}_{d\ell=0} \frac{\partial'(\ell(z)dz) \, \partial''(\ell(z)dz) \, \partial'''(\ell(z)dz)}{d\ell(z)dz}
 \]
 via the relation
 \[
 \langle \partial' \circ \partial'', \partial''' \rangle_{\eta} = c(\partial', \partial'', \partial'''),
 \]
 where \( \partial', \partial'', \partial''' \in \operatorname{Vect}(M) \).

The proof of Corollary \ref{maincor4} follows from a direct comparison of the explicit formulas for the Frobenius structures of $M$ and $\mathcal{M}$. As shown in \cite{ma2023principal}, the Levi-Civita connection on $M$ is given by:
\begin{equation}\label{cKPcon}
	(\nabla_{\partial_1} \partial_2)(\ell(z)) = \partial_1 \partial_2 \ell(z) - \biggl(\frac{\partial_1 \ell(z) \partial_2 \ell(z)}{\ell'(z)}\biggr)_{\!\infty,\,\ge 0}'
	- \sum_{j=1}^{m}\biggl(\frac{\partial_1 \ell(z) \partial_2 \ell(z)}{\ell'(z)}\biggr)_{\!\varphi_{j},\,\le -1}'.
\end{equation}
A direct verification shows that this expression is precisely the limit of the connection on $\mathcal{M}$, given by equation \eqref{conn}, as $\zeta(z)\to 0$. 

Furthermore, the dual description of the multiplication for $M$ is described as follows. We represent a covector in $T^{\ast}_{\ell(z)}M$ by an element $\omega(z)\in\mathcal{H}$ via the pairing:
\begin{equation}
	\langle \omega(z), X \rangle := \frac{1}{2\pi \mathrm{i}} \sum_{j=1}^{m} \oint_{\gamma_j} \omega(z) \partial_X \ell(z) \, dz.
\end{equation}
The linear map \( \eta: \mathcal{H} \to T_{\ell(z)}M \) is then given by
\begin{equation}\label{kdvdulmet}
	\eta(\omega(z)) = -(\omega(z))_{+}\ell'(z) + (\omega(z)\ell'(z))_{+}
\end{equation}
and it relates to the metric via the identity
\begin{equation}
	\langle\omega(z), X\rangle = \langle\eta(\omega(z)), X\rangle_{\eta}.
\end{equation}
The multiplication operator $C_{X}: T^{\ast}_{\ell(z)}M \to T_{\ell(z)}M$, indexed by a tangent vector $X \in T_{\ell(z)}M$, is defined by the formula
\begin{equation}\label{kdvdulpro}
	C_{X}(\omega(z)) = -(\omega(z)\partial_{X}\ell(z))_{+}\ell'(z) + (\omega(z)\ell'(z))_{+}\partial_{X}\ell(z),
\end{equation}
and is related to the tangent space multiplication via
\begin{equation}
	C_{X}(\omega(z)) = X \circ \eta(\omega(z)).
\end{equation}
	
A direct comparison shows that in the limit as $\zeta(z)\to 0$, the Levi-Civita connection $\nabla$ and the multiplication operator $C_{X}$ for $\mathcal{M}$ reduce precisely to their counterparts on $M$. Since the proofs of Theorem \ref{mainthm1} and Theorem \ref{mainthm2} rely solely on these structures, their arguments remain valid in this limit. Corollary \ref{maincor4} is thus established. \hfill \qed

\subsection{Proof of Corollaries \ref{maincor5} and \ref{maincor6}}
The proofs of these two corollaries are analogous to that of Corollary \ref{maincor4}. Recall that $\hat{\mathcal{M}}$ is defined as a submanifold of $\mathcal{M}$ by imposing the following restrictions on the flat coordinates:
\[
t_{1,2l}=0, \quad t_{2i-2,l}=-t_{2i-1,l}, \quad 2\le i\le m'+1, \quad l\in\mathbb{Z},
\]
and
\begin{align*}
	h_{0,2j} &= 0, \quad j=1, \dots, n_{0}'; \\[5pt]
	h_{2k-2,r} &= -h_{2k-1,r}, \quad k=2, \dots, m'+1, \quad r=0, \dots, n_{k}',
\end{align*}
which is equivalent to
\[
a(z)=a(-z), \quad \hat{a}(z)=\hat{a}(-z)
\]
for $\vec{a}=(a(z),\hat{a}(z))\in \hat{\mathcal{M}}$. One can verify that the Levi-Civita connection $\nabla$ associated with the restricted metric on $\hat{\mathcal{M}}$ retains the same explicit form as that on $\mathcal{M}$ (see \eqref{conn}).

To provide a dual description of this metric, let $\mathcal{H}^{\mathrm{odd}} \subset \mathcal{H}$ be the subspace of functions satisfying $\omega(-z) = -\omega(z)$. Since the projection operators $(\cdot)_{\pm}$ commute with the involution $\iota \colon z \to -z$, we have
\[
-\omega_{\pm}(z)=(\omega(-z))_{\pm}=\omega_{\pm}(-z)
\]
for any $\omega(z) \in \mathcal{H}^{\mathrm{odd}}$. Consequently, the linear map $\eta \colon \mathcal{H} \times \mathcal{H} \to T_{\vec{a}}\mathcal{M}$ defined in \eqref{etaom} restricts to a map from $\mathcal{H}^{\mathrm{odd}}\times \mathcal{H}^{\mathrm{odd}}$ to $T_{\vec{a}}\hat{\mathcal{M}}$.

Using the explicit formula \eqref{mulope} for the multiplication operator $C_{\vec{\xi}}$, it follows that for any $\vec{\xi}\in T_{\vec{a}}\hat{\mathcal{M}}$, the operator $C_{\vec{\xi}}$ restricts to a map from $\mathcal{H}^{\mathrm{odd}} \times \mathcal{H}^{\mathrm{odd}}$ to $T_{\vec{a}}\hat{\mathcal{M}}$.

Thus, following the construction in \cite{ma2024infinite}, we can analogously establish a Frobenius manifold structure on $\hat{\mathcal{M}}$, which can be viewed as the restriction of the structure on $\mathcal{M}$. Furthermore, adopting the approach in \cite{ma2024infinite}, the principal hierarchy for $\hat{\mathcal{M}}$ is given by the following subhierarchy of that for $\mathcal{M}$:
\begin{gather*}
	\left\{\frac{\partial}{\partial T^{t_{1,2l-1}}}\right\}_{l\in\mathbb{Z}} \cup \left\{\frac{\partial}{\partial T^{t_{2i-2,l}}}-\frac{\partial}{\partial T^{t_{2i-1,l}}}\right\}_{2\le i\le m'+1;\, l\in\mathbb{Z}} \\
	\text{and} \\
	\left\{\frac{\partial}{\partial T^{h_{0,2j-1}}}\right\}_{1\le j\le n_{0}'} \cup \left\{\frac{\partial}{\partial T^{h_{1,2j-1}}}\right\}_{1\le j\le n_{1}'} \cup \left\{\frac{\partial}{\partial T^{h_{2k-2,r}}}-\frac{\partial}{\partial T^{h_{2k-1,r}}}\right\}_{2\le k\le m'+1;\, 0\le r\le n_{k}'}.
\end{gather*}
Finally, the solutions to the open WDVV equations and the associated principal hierarchy for $\hat{\mathcal{M}}$ are obtained directly by restricting those associated with $\mathcal{M}$. This completes the proof of Corollary \ref{maincor5}.

The proof of Corollary \ref{maincor6} proceeds analogously to that of Corollary \ref{maincor4}. It relies on a direct comparison of the explicit forms of the connection and multiplication operator on $\hat{\mathcal{M}}$ with those on $\hat{M}$ (as detailed in \cite{ma2023principal}). Therefore, we omit the details for brevity. \hfill \qed

\medskip
{\small
	\noindent{\bf Acknowledgements}.
The author thanks Professors Chao-Zhong Wu and Dafeng Zuo for their stimulating discussions and valuable suggestions.

\bibliographystyle{unsrt}

\begin{thebibliography}{10}
\bibitem{dub1998}
Dubrovin B.
\newblock Geometry of 2d topological field theories.
\newblock In {\em Integrable Systems and Quantum Groups}, pages 120--348.
Springer, 1998.

\bibitem{manin1999frobenius}
Manin I.
\newblock {\em Frobenius manifolds, quantum cohomology, and moduli spaces},
volume~47.
\newblock American Mathematical Soc., 1999.

\bibitem{hertling2002frobenius}
Hertling C.
\newblock {\em Frobenius manifolds and moduli spaces for singularities}, volume
151.
\newblock Cambridge University Press, 2002.

\bibitem{dubrovin2001normal}
Dubrovin B. and Zhang Y.
\newblock Normal forms of hierarchies of integrable \text{PDE}s,
\text{F}robenius manifolds and \text{Gromov-Witten} invariants.
\newblock {\em arXiv preprint math/0108160}, 2001.

\bibitem{witten1990two}
Witten E.
\newblock Two-dimensional gravity and intersection theory on moduli space.
\newblock {\em Surveys in differential geometry}, 1(1):243--310, 1990.

\bibitem{kontsevich1992intersection}
Kontsevich M.
\newblock Intersection theory on the moduli space of curves and the matrix
\text{A}iry function.
\newblock {\em Communications in Mathematical Physics}, 147:1--23, 1992.

\bibitem{liu2015bcfg}
Liu S.-Q., Ruan Y., and Zhang Y.
\newblock \text{BCFG Drinfeld--Sokolov hierarchies and FJRW-theory}.
\newblock {\em Inventiones mathematicae}, 201:711--772, 2015.

\bibitem{dubrovin2016hodge}
Dubrovin B., Liu S.-Q., Yang D., and Zhang Y.
\newblock Hodge integrals and tau-symmetric integrable hierarchies of
\text{Hamiltonian evolutionary PDEs}.
\newblock {\em Advances in Mathematics}, 293:382--435, 2016.

\bibitem{liu2022variational}
Liu S.-Q., Wang Z., and Zhang Y.
\newblock Variational bihamiltonian cohomologies and integrable hierarchies
\text{II}: Virasoro symmetries.
\newblock {\em Communications in Mathematical Physics}, 395(1):459--519, 2022.

\bibitem{liu2023variational}
Liu S.-Q., Wang Z., and Zhang Y.
\newblock Variational bihamiltonian cohomologies and integrable hierarchies
\text{I}: foundations.
\newblock {\em Communications in Mathematical Physics}, 401(1):985--1031, 2023.

\bibitem{liu2025generalized}
Liu S.-Q., Qu H., and Zhang Y.
\newblock Generalized Frobenius manifolds with non-flat unity and integrable
hierarchies.
\newblock {\em Communications in Mathematical Physics}, 406(4):77, 2025.

\bibitem{carlet2011infinite}
Carlet G., Dubrovin B., and Mertens~L. P.
\newblock Infinite-dimensional Frobenius manifolds for $2+1$ integrable systems.
\newblock {\em Mathematische Annalen}, 349:75--115, 2011.

\bibitem{wu2012class}
Wu~C.-Z. and Xu~D.
\newblock A class of infinite-dimensional Frobenius manifolds and their
submanifolds.
\newblock {\em International Mathematics Research Notices},
2012(19):4520--4562, 2012.

\bibitem{wu2014infinite}
Wu~C.-Z. and Zuo D.
\newblock Infinite-dimensional Frobenius manifolds underlying the Toda
lattice hierarchy.
\newblock {\em Advances in Mathematics}, 255:487--524, 2014.

\bibitem{ma2021infinite}
Ma~S., Wu~C.-Z., and Zuo D.
\newblock Infinite-dimensional Frobenius manifolds underlying an extension of
the dispersionless Kadomtsev--Petviashvili hierarchy.
\newblock {\em Journal of Geometry and Physics}, 161:104006, 2021.

\bibitem{raimondo2012frobenius}
Raimondo A.
\newblock Frobenius manifold for the dispersionless Kadomtsev-Petviashvili
equation.
\newblock {\em Communications in Mathematical Physics}, 311:557--594, 2012.

\bibitem{szablikowski2015classical}
Szablikowski B. M.
\newblock Classical r-matrix like approach to Frobenius manifolds, WDVV
equations and flat metrics.
\newblock {\em Journal of Physics A: Mathematical and Theoretical},
48(31):315203, 2015.

\bibitem{carlet2015principal}
Carlet G. and Mertens~L. P.
\newblock Principal hierarchies of infinite-dimensional Frobenius manifolds:
The extended 2d Toda lattice.
\newblock {\em Advances in Mathematics}, 278:137--181, 2015.

\bibitem{krichever1988method}
Krichever~I. M.
\newblock Method of averaging for two-dimensional ``integrable'' equations.
\newblock {\em Functional Analysis and Its Applications}, 22(3):200--213, 1988.

\bibitem{krichever1994tau}
Krichever~I. M.
\newblock The $\tau$-function of the universal Whitham hierarchy, matrix
models and topological field theories.
\newblock {\em Communications on Pure and Applied Mathematics}, 47(4):437--475,
1994.

\bibitem{gorsky1995integrability}
Gorsky A., Krichever I., Marshakov A., Mironov A., and Morozov A.
\newblock Integrability and Seiberg-Witten exact solution.
\newblock {\em Physics Letters B}, 355(3-4):466--474, 1995.

\bibitem{itoyama1996integrability}
Itoyama H. and Morozov A.
\newblock Integrability and Seiberg-Witten theory curves and periods.
\newblock {\em Nuclear Physics B}, 477(3):855--877, 1996.

\bibitem{itoyama1997prepotential}
Itoyama H. and Morozov A.
\newblock Prepotential and the Seiberg-Witten theory.
\newblock {\em Nuclear Physics B}, 491(3):529--573, 1997.

\bibitem{krichever1997integrable}
Krichever~I. M. and Phong~D. H.
\newblock On the integrable geometry of soliton equations and $n=2$
supersymmetric gauge theories.
\newblock {\em Journal of Differential Geometry}, 45(2):349--389, 1997.

\bibitem{gorsky1998rg}
Gorsky A., Marshakov A., Mironov A., and Morozov A.
\newblock RG equations from Whitham hierarchy.
\newblock {\em Nuclear Physics B}, 527(3):690--716, 1998.

\bibitem{krichever2004laplacian}
Krichever I., Mineev-Weinstein M., Wiegmann P., and Zabrodin A.
\newblock Laplacian growth and Whitham equations of soliton theory.
\newblock {\em Physica D: Nonlinear Phenomena}, 198(1-2):1--28, 2004.

\bibitem{zabrodin2005whitham}
Zabrodin A.
\newblock Whitham hierarchy in growth problems.
\newblock {\em Theoretical and Mathematical Physics}, 142(2):166--182, 2005.

\bibitem{wiegmann2000conformal}
Wiegmann~P. B. and Zabrodin A.
\newblock Conformal maps and integrable hierarchies.
\newblock {\em Communications in Mathematical Physics}, 213:523--538, 2000.

\bibitem{martinez2006genus}
Alonso~L. M. and Medina E.
\newblock Genus-zero Whitham hierarchies in conformal-map dynamics.
\newblock {\em Physics Letters B}, 641(6):466--473, 2006.

\bibitem{ma2024infinite}
Ma~S., Wu~C.-Z., and Zuo D.
\newblock Infinite-dimensional Frobenius Manifolds and extensions of genus-zero Whitham hierarchies.
\newblock {\em arXiv preprint arXiv:2406.08239}, 2024.

\bibitem{horev2012open}
Horev A. and Solomon~J. P.
\newblock The open Gromov-Witten-Welschinger theory of blowups of the
projective plane.
\newblock {\em arXiv preprint arXiv:1210.4034}, 2012.

\bibitem{pandharipande2024intersection}
Pandharipande R., Solomon~J. P., and Tessler R. J.
\newblock Intersection theory on moduli of disks, open KdV and Virasoro.
\newblock {\em Geometry \& Topology}, 28(6):2483--2567, 2024.

\bibitem{buryak2022open}
Buryak A., Clader E., and Tessler R. J.
\newblock Open $r$-spin theory I: Foundations.
\newblock {\em International Mathematics Research Notices},
2022(14):10458--10532, 2022.

\bibitem{buryak2024open}
Buryak A., Clader E., and Tessler R. J.
\newblock Open $r$-spin theory II: the analogue of Witten's conjecture for
$r$-spin disks.
\newblock {\em Journal of Differential Geometry}, 128(1):1--75, 2024.

\bibitem{buryak2015equivalence}
Buryak A.
\newblock Equivalence of the open KdV and the open Virasoro equations for the
moduli space of Riemann surfaces with boundary.
\newblock {\em Letters in Mathematical Physics}, 105:1427--1448, 2015.

\bibitem{Buryak:2018ypm}
Buryak A., Clader E., and Tessler~R. J.
\newblock Closed extended $r$-spin theory and the Gelfand–Dickey wave function.
\newblock {\em Journal of Geometry and Physics}, 137: 132--153, 2019.

\bibitem{basalaev2019open}
Basalaev A. and Buryak A.
\newblock Open WDVV equations and Virasoro constraints.
\newblock {\em Arnold Mathematical Journal}, 5(2):145--186, 2019.

\bibitem{arsie2021flat}
Arsie A., Buryak A., Lorenzoni P., and Rossi P.
\newblock Flat F-manifolds, F-CohFTs, and integrable hierarchies.
\newblock {\em Communications in Mathematical Physics}, 388(1):291--328, 2021.

\bibitem{arsie2023semisimple}
Arsie A., Buryak A., Lorenzoni P., and Rossi P.
\newblock Semisimple flat F-manifolds in higher genus.
\newblock {\em Communications in Mathematical Physics}, 397(1):141--197, 2023.

\bibitem{basalaev2021open}
Basalaev A. and Buryak A.
\newblock Open Saito theory for A and D singularities.
\newblock {\em International Mathematics Research Notices}, 2021(7):5460--5491,
2021.

\bibitem{aoyama1996topological}
Aoyama S. and Kodama Y.
\newblock Topological Landau-Ginzburg theory with a rational potential and the
dispersionless KP hierarchy.
\newblock {\em Communications in Mathematical Physics}, 182:185--219, 1996.

\bibitem{liu2015central}
Liu S.-Q., Zhang Y., and Zhou X.
\newblock Central invariants of the constrained \text{KP} hierarchies.
\newblock {\em Journal of Geometry and Physics}, 97:177--189, 2015.

\bibitem{yang2021extension}
Yang D. and Zhou C.
\newblock On an extension of the generalized BGW tau-function.
\newblock {\em Letters in Mathematical Physics}, 111(5):123, 2021.

\bibitem{zuo2007frobenius}
Zuo D.
\newblock Frobenius manifolds associated to $B_{l}$ and $D_{l}$, revisited.
\newblock {\em International Mathematics Research Notices}, 2007:rnm020, 2007.

\bibitem{liu2011drinfeld}
Liu S.-Q., Wu~C.-Z., and Zhang Y.
\newblock On the Drinfeld--Sokolov hierarchies of D type.
\newblock {\em International Mathematics Research Notices}, 2011(8):1952--1996,
2011.

\bibitem{drinfel1985lie}
Drinfel'd~V. G. and Sokolov~V. V.
\newblock Lie algebras and equations of Korteweg-de Vries type.
\newblock {\em Journal of Soviet Mathematics}, 30:1975--2036, 1985.

\bibitem{manin2005f}
Manin~Y. I.
\newblock F-manifolds with flat structure and Dubrovin's duality.
\newblock {\em Advances in Mathematics}, 198(1):5--26, 2005.

\bibitem{arsie2018flat}
Arsie A. and Lorenzoni P.
\newblock Flat F-manifolds, Miura invariants, and integrable systems of
conservation laws.
\newblock {\em Journal of Integrable Systems}, 3(1):xyy004, 2018.

\bibitem{alcolado2017extended}
Alcolado A.
\newblock {\em Extended Frobenius manifolds and the open WDVV equations}.
\newblock McGill University (Canada), 2017.

\bibitem{ma2023principal}
Ma S.
\newblock Principal hierarchy of Frobenius manifolds associated with rational
and trigonometric Landau-Ginzburg superpotentials.
\newblock {\em arXiv preprint arXiv:2311.17363}, 2023.




\end{thebibliography}

\end{document}